\documentclass[sigconf,timestamp=false]{acmart}
\usepackage{caption}
\usepackage{subcaption}
\usepackage{algorithmic}
\usepackage{algorithm2e}[linesnumbered,ruled,vlined]
\usepackage{colortbl}

\setcopyright{acmcopyright}

\acmConference[Arxiv'22] {Advertising Conference}{2022}{Silicon Valley, CA, USA}

\begin{document}
\title{Demystifying Advertising Campaign Bid Recommendation: A Constraint target CPA Goal Optimization} 

\author{{Deguang Kong}{*}, Konstantin Shmakov and Jian Yang} 
\authornote{This work was done during the authors were working at Yahoo Research.  The
correspondence should be addressed to doogkong@gmail.com.  The views and conclusions contained in
this document are those of the author(s) and should not be
interpreted as representing the official policies, either
expressed or implied, of any companies. }
\affiliation{%
  \institution{Yahoo Research,   San Jose, California, U.S.A, 94089}
  }

  \email{ doogkong@gmail.com, kshmakov@yahooinc.com,  jianyang@yahooinc.com}

\begin{abstract}
In cost-per-click (CPC) or cost-per-impression (CPM) advertising campaigns, advertisers always run the risk of spending the budget without getting enough conversions. Moreover, the bidding on advertising inventory has few connections with propensity one can reach to target cost-per-acquisition (tCPA) goals. To address this problem, this paper 
presents a bid optimization scenario to achieve the desired CPA goals for advertisers. In particular, we build the optimization engine to make a decision by solving the rigorously formalized constrained optimization problem, which leverages the bid landscape model learned from rich historical auction data using non-parametric learning.  The proposed model can naturally recommend the bid that meets the advertisers' expectations by making inference over advertisers' history auction behaviors, which essentially deals with the data challenges commonly faced by bid landscape modeling: incomplete logs in auctions, and uncertainty due to the variation and fluctuations in advertising bidding behaviors.  The bid optimization model outperforms the baseline methods on real-world campaigns, and has been applied into a wide range of scenarios for performance improvement and revenue liftup. 
\end{abstract}

%
%

\maketitle

\section{introduction}

Online advertising\footnote{\url{https://www.statista.com/statistics/183816/us-online-advertising-revenue-since-2000/}} uses the internet to deliver the promotional marketing messages provided by advertisers to consumers on the publisher platform by integrating the advertisements\footnote{ {``Ad'' is a shortened colloquial form of  ``advertisement''.}} into its online content.  
Bid optimization is a technique to optimize the bidding process, which allows the advertiser to maximize the revenue, return on investment (ROI) and other performance metrics. Native advertisement appeared in publisher side where the ad experience looks like the natural without affecting the function of user experiences.  An example  of native ad is shown in Fig.~\ref{fig:example}.

The advertising platform we built for a major web portal has its unique ecosystem, which owns both the demand side and supply side platforms that allow us to collect
integrated supply side impressions, and ad campaign information on demand side.  
Given the rich campaign history data, learning the advertisers' bidding and targeting behaviors 
is very prominent to forecast the advertisers' future behaviors so as to make informed decisions for bidding and budget planning.    The traditional way of manual bidding or rule-based bidding (i.e., generation of bidding rules) 
is not trustful because the problem complexity is steadily increasing. For example, for the same key word ``buy jeans", there is 20 conversions with
\$2 CPC bid for advertiser 1 but only 5 conversions with the same price for advertiser 2. To effectively learn the useful 
"signals" available for bid adjustment to advertisers, we build effective landscape models to capture what the cost and win rate will look like if the advertisers change the bidding price.   For performance campaigns (those that link advertising to a user's purchase, a sign-up, or some other desired action), we require the learning algorithms to leverage history data into the
buying process through a learning and optimization process which enhances an advertiser's natural instincts by automatically optimizing bids based on historical observations.

\begin{figure}[t]
	\centering
	\includegraphics[height=1.7in,width=0.45\textwidth]{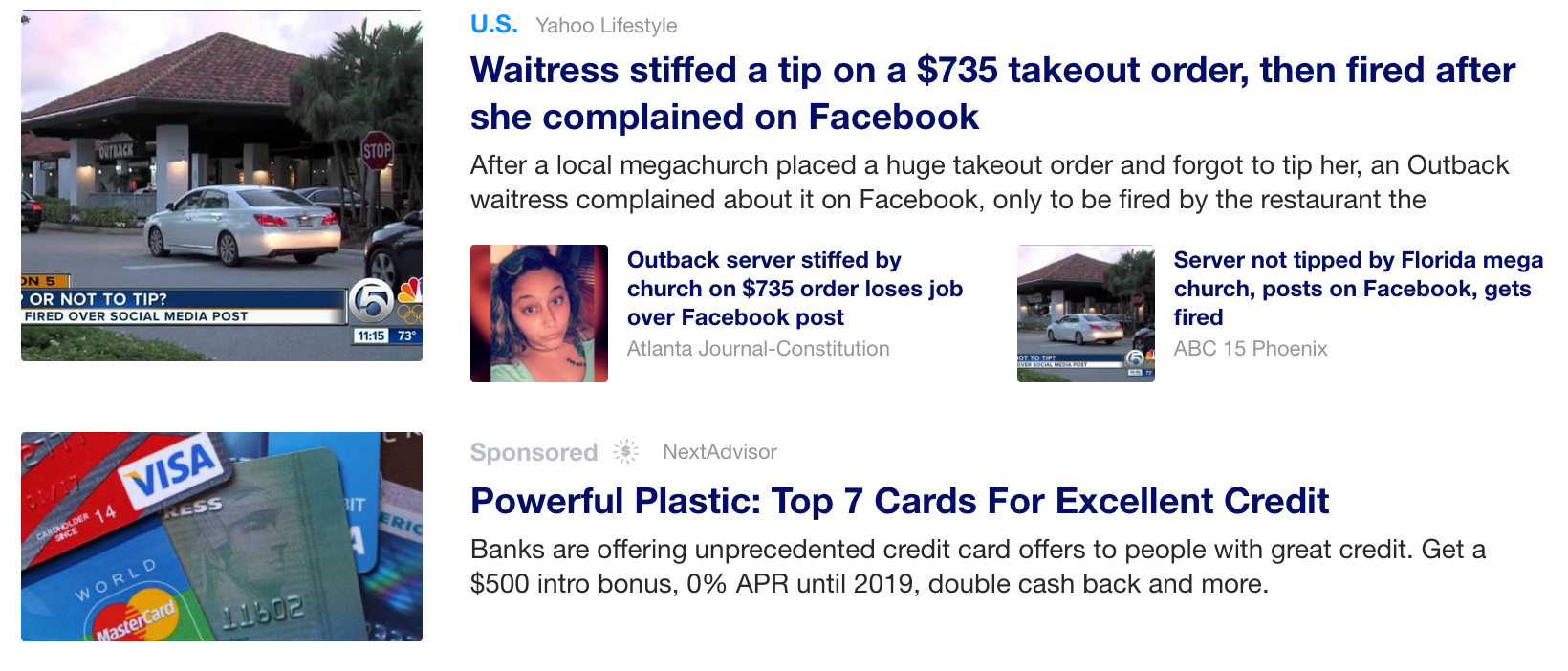}
		\caption{\small{The native ad appears in web display pages. The sponsored ad of  ``visa credit'' card looks like the news.  
	}
	}
	\label{fig:example}
\end{figure}

While advertisers may require to achieve desired goals (such as  clicks, conversions or ROI),  we believe cost per acquisition (CPA) goal is always important for consideration of the transformation of clicks into the true conversions, especially given the fact that many advertisers are more concerned with the increase in revenue.  In this paper, we propose effective ways for CPA goal optimization\footnote{
As the social media ads, Facebook actually used optimized cost per mille (oCPM) criterion to allow advertisers to bid for click and actually pay per impression.  Cost per mille impressions (CPM) poses high risk on the conversions, 
and cost per click (CPC) controls the cost of clicks after running page views without giving enough attention to the true conversions that may actually bring in revenue increases. 
},  including achieving the corresponding conversion and sales (using target CPA goals). 
Our approach will help the advertisers programmatically choose the best ``smart'' bid, and then derive the most possible scale for the targets, which in turn, acquires more traffic for e-commerce business. 

In the typical advertising system,  advertising platform logs all campaign history auction events (BH module in Fig.\ref{fig:flow}), and then bid landscape model and key performance indicator (KPI) modules (BLP  in Fig.\ref{fig:flow}) such as click-through-rate (CTR)~\cite{kongctr}, conversion rate (CVR) prediction models are built using history observations.  After that the bid optimization strategy module (BOPT in Fig.\ref{fig:flow}) is responsible for providing the optimized bidding strategy to advertisers. 
This paper presents a strategy for bid optimization to perform bid adjustment based on how likely a click is to lead to a conversion. It therefore can help advertisers to plan CPA goals and revenue targeting. The proposed bid optimization strategy plays the important role to the following use cases (shown in Fig.\ref{fig:flow}): (i)An advertiser wants to create a new advertising campaign and aims to find the answers to the questions such as ``\emph{what price should I bid  to achieve the best performance?}" After running bid optimization module, we can provide bid suggestions to the advertiser. 
(ii) For real-time ad serving events on Demand side,  the advertiser wants to know "\emph{how much should I bid to win ad inventory on publishers?}" After running the bid optimization module, we provide the bid recommendations for advertisers to win impressions. 
\begin{figure*}[t]
	\centering
	\includegraphics[height=1.8in,width=0.8\linewidth]{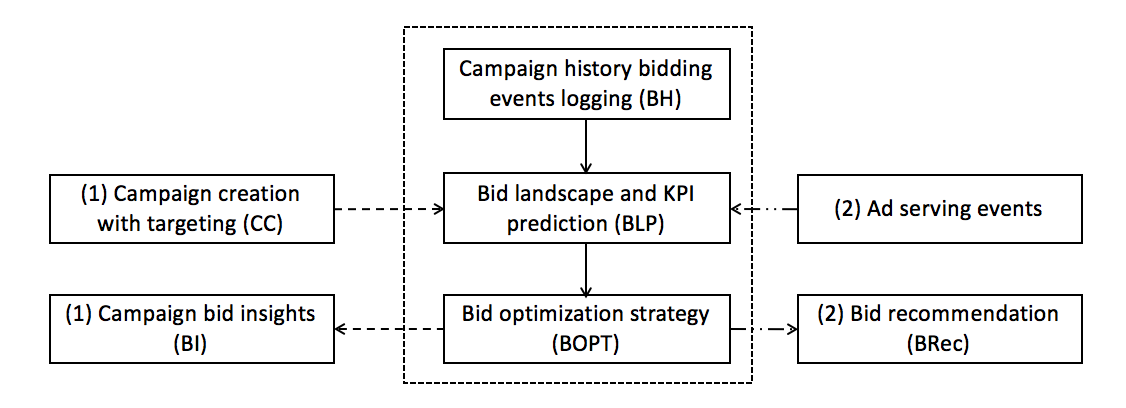}
	\caption{\small Workflow of bid optimization in advertising system.}
	\label{fig:flow}
\end{figure*}
To summarize, this work makes the following contributions. 

$\bullet$ To help advertisers achieve the desired conversion goals, we propose a \emph{constraint optimization} scheme that allows the advertisers to smartly bid by naturally transforming the CPA targeting to bid. 

$\bullet$ To overcome the limitation of data sparsity and uncertain auction behaviors in ad marketplace, we propose an effective \emph{bid landscape model} by learning the win-rate and eCPM cost distributions from history data that can well capture the advertiser auction behavior variations and fluctuations. 

$\bullet$ Extensive experiment results demonstrate the \emph{performance improvement} over the baseline methods. A/B test shows our algorithm yields better results than current bid without optimization. We deliver and validate our approach in online advertising system in production. 


\section{Constraint CPA Goal Optimization}

\subsection{Preliminary and Related Works}

In the ad ecosystem,  
publishers would like to keep site profitable by selling ads to advertisers and provide display opportunity (on web pages, social media, etc). The advertisers would like to win these opportunities through participation of auctions.  Ad exchange is responsible for making communications between supply side (publisher) and demand side (advertisers) by allowing auctions for display opportunities.  Most popular auction mechanism for multiple items is generalized second-price auction (GSP)~\cite{10.1257/aer.97.1.242}~\cite{DBLP:conf/sigecom/AggarwalGM06}, where after each bidder places a bid, the highest bidder gets the first slot and the second-highest wins the second slot and so on, but the highest bidder only needs to pay the price bid by the second-highest bidder.  In practice, Cavallo et al. \cite{Cavallo:2017:SSA:3038912.3052703} re-designed the search ads selling system that achieves the properties analogous to those held by GSP for supporting a variety of ads sizes, decorations, and other distinguishing feature. 

In online ads system, one common objective is to maximize the advertisers' value~\cite{DBLP:conf/kdd/ZhuJTPZLG17}. Predicting click-through-rate (CTR) and conversion rate (CVR) are critical tasks for online advertising. For example, Edizel et al. \cite{Edizel2017} proposed to use convolutional neural networks to predict ad-level CTR for sponsored search\footnote{\small{Please refer to \S 2.4 for more discussions.}},  Lee et al. \cite{Lee:2012:ECR:2339530.2339651} proposed a logistic regression model for conversion modeling,  and pre-click and post-click user engagement are improved in ~\cite{Zhou:2016:PPQ:2872427.2883053}, ~\cite{Barbieri:2016:IPU:2872427.2883092}, respectively.

Bid optimization~\cite{Perlich:2012:BOI:2339530.2339655} ~\cite{Zhang:2014:ORB:2623330.2623633} is a process to recommend the best bid for the advertisers. On one hand, if the advertisers bid too low, it may lose the ads display opportunity due to GSP mechanism. On the other hand, if one bids too high, it will waste money but get less leverage on clicks or other KPIs.  Bid optimization aims to provide the advertiser the best bid price to achieve the desired goal from learning the history campaign performances.  
Cost per acquisition (CPA) measures how much the advertiser needs to pay for per acquisition (e.g., conversion), and allows one to get best performance out of the campaigns with target CPA. This is useful because the advertiser can start with targets that align with historical CPA.   Further, CPA goal optimization predicts future conversions associated with the reported conversion values (through conversion tracking).  
In Google Adwords, one of the bidding strategies is target CPA \footnote{\url{ https://support.google.com/adwords/answer/6268632?hl=en}}, where the conversion can be tracked using Google Analytics for monetization in online advertising. 
for Google Display Network, Search and Youtube. 

Recently, a horizon-independent pricing approach~\cite{DBLP:conf/www/Drutsa17} is proposed for the ad auctions in which the seller does not know in advance the number of played rounds, and Mirrokni et al. \cite{DBLP:conf/www/MirrokniN17} designed the new contract for advertisers with advanced targeting capabilities in order to maximize the revenue.  Budget management is considered in ad auctions by combining the seller's profit and buyers' utility \cite{DBLP:conf/www/BalseiroKMM17}. Many works discuss price design, including reserve prices design in single item auctions~\cite{DBLP:conf/www/LemePV16},  mixture bidder of GSP and Vickrey-Clarke-Groves (VCG)~\cite{Bachrach:2016:MDM:2872427.2882983}, algorithmic/dynamic pricing~\cite{Chen:2016:EAA:2872427.2883089}, 
joint optimization of bid and budget allocation in sponsored search~\cite{Zhang:2012:JOB:2339530.2339716}, optimal real-time bidding for display advertising~\cite{Zhang:2014:ORB:2623330.2623633},  budget constrained bidding in keyword auctions and online knapsack problems~\cite{Zhou:2008:BCB:1504941.1505019}, {\it etc}. The exploration-exploitation strategy has been used in contextual multi-armed bandit settings for  content recommendation~\cite{Li:2016:CFB:2911451.2911548}, with further extension to  cluster-of-bandit algorithm~\cite{DBLP:conf/icml/GentileLKKZE17}  for collaborative recommendation tasks, and linear bandit algorithm~\cite{DBLP:conf/icml/KordaSL16}. Orthogonal to the above works~\cite{DBLP:conf/www/GaoKLBY18},~\cite{DBLP:conf/www/KongSY18a}, ~\cite{DBLP:conf/www/KongFSY18}, this paper presents the constraint bid optimization 
for the best conversion performance while achieving the CPA goals. The short version of this work appeared in ~\cite{DBLP:conf/www/KongSY18}.


\subsection{Our approach} 
{\bf A motivating example} Let's say the advertiser is measuring sales for several online electronic stores by optimizing the bids based on the value of a shopping cart in total.  The goal is 50 dollars per conversion (that leads to sales) based on every dollar one spends on ads. If one is setting a target CPA goal as \$50 dollars- for every \$50 you spend on ads, you'd like to get 1 conversion (and get the corresponding revenue).  More concretely, 

Store 1: \$100 in spend $\div$ 2 conversions  = \$50  CPA goal

Store 2: \$150 in spend $\div$ 4 conversions  = \$37.5  CPA goal

Store 3:  \$500 in spend $\div$ 5 conversions  = \$100 CPA goal

Based on the data one would decide to spend more budget at store 1 and store 2, as one is hitting the goal on these websites. At store 3, the advertiser pays the double of the goal per conversion, so the campaign will stop bidding there. 

{\bf Optimization Goal} A question that naturally follows is:  \emph{how to set the bid price in order to achieve the desired CPA goals?} More formally, given the CPA goal with value $C$, we aim to provide the bid recommendation that can maximize the conversions while considering CPA goal for the advertiser $i$, {\it i.e.,}
\begin{eqnarray}
& \max_{\text{bid}}Conversion(bid);  \nonumber \\
& s.t \;\;\;\; CPA(bid) = C; \;\; \;  Spend(bid) \leq B 
\label{EQ:click_CPA}
\end{eqnarray}
where $Conversion(bid)$ denotes the number of conversions the advertiser can obtain at price $bid$, $CPA(bid)$ tells the CPA goal one is supposed to achieve at price $bid$ and $spend(bid)$ limits the true cost of advertisers\footnote{\small {For notation simplicity purpose, advertising campaign $i$ is ignored in the paper next, and all the modeling is applied for the particular campaign $i$. }} 
within the budget $B$.   An alternative approach is to maximize the revenue by considering the factor of conversion value from eliminating the bad advertising inventory, which, however can be similarly solved as that in Eq.(\ref{EQ:click_CPA}). Through this optimization, the inventory that leads to higher conversions is ranked higher, and recommended to advertisers.
We run the campaign analysis to solve Eq.(\ref{EQ:click_CPA}) using O\&O (owned and operated) advertising platform. The goal is to achieve the superior reach, find the correct targeting and perform optimization.

\section{Optimization Workflow}

We provide the solution to Eq.(\ref{EQ:click_CPA}) using machine learning and optimization. 

{\bf Step1: Get all supply impressions and predict click and conversions}  %
The best performing ad space for the ad campaign is associated with the characteristics that influence the type of users they attract, including publisher and URL (a.k.a supply group) and geo-locations. Essentially, in most performance buying scenarios, number of impression we can get from the inventory is determined by a combination of the content of the site and a specific campaign's creative and targeting, i.e., the number of impressions one can get from the supply side is proportional to the available number of impressions on the particular publisher, goe location and targeting attribution settings. Therefore, the number of impressions is given by querying the total supplies with pre-defined parameter combinations, such as geo-location, supply group, age, gender, etc.

In ads serving environment, impressions are most commonly transacted between advertisers, ad exchange and publishers on a cost-per-thousand impressions (CPM) basis. Since CPC/CPA is measured differently than CPM, 
translating Goals of CPA into click and conversion number prediction is necessary for advertisers, because advertisers typically pay for a particular user response via a cost-per-click (CPC) or CPM.  More formally, 
\begin{eqnarray}
Click(bid) = Impression \times Win rate(bid) \times pCTR, 
\label{EQ:click}
\end{eqnarray}
\begin{eqnarray}
Conversion(bid) = Click(bid) \times pCVR 
\label{EQ:conversion}
\end{eqnarray}
where $Click(bid)$ tells the number of clicks the advertiser can obtain at $bid$ price,  $Impression$ is the available impressions based on the combination of publisher and targeting, $Win rate(bid)$ is the winning ratio from all available impressions,  $Conversion(bid)$ is the number of conversions obtained at price $bid$,  $pCTR, pCVR$ are the predicted click-through-rate (CTR) and predicted conversion-rate (CVR) for the particular advertiser based on history campaigns.


{\bf Step 2: CPA optimization for bid}
Now we are ready to translate the CPA goals into estimated CPC bids based
on a variety of factors, including historical win rates for a particular type of
impression, and actual cost at that bid price.   For example, if the historical rates indicate that the win rate is low, then a lower bid will be submitted.  On the other hand, if the response rate is high, then a higher bid will be submitted. 

\begin{lemma}
\begin{eqnarray}
CPA(bid)= \frac{eCPM\_cost(bid)}{1000 \times pCTR \times pCVR},
\label{EQ:CPA}
\end{eqnarray}
where $Spend(bid)$ is the money spent,  $CPC\_cost$ is the cost-per-click,  and 
$eCPM\_cost$ is the cost-per-mille at  price $bid$, respectively.
\end{lemma}

\begin{proof}
\begin{eqnarray}
CPA(bid) &=& \frac{Spend(bid)}{Conversion(bid)} = \frac{Click(bid) \times CPC\_cost(bid)}{Click(bid) \times pCVR} \nonumber \\
&=&  \frac{CPC\_cost(bid)}{pCVR}
=  \frac{eCPM\_cost(bid)}{1000 \times pCTR \times pCVR} .
\end{eqnarray}
\end{proof}

For example, if the advertiser is willing to spend \$3 per conversion and the likelihood of wining clicks and conversions is 0.01\%. Since advertising inventory is priced in cost-per-thousand impressions, the CPM winning cost should be \$0.30. If the probability of a response is higher, the
bid is increased.  High CPM impressions are worth every penny if they result in converting users for a
lower CPA/CPC.  
It is not the cost of the impression that matters, instead the cost of the acquisition does matter.  Please note when the spend of auction exceeds the budget, even if the advertiser
bids more, it is not necessarily increasing cost due to the inventory capacity. 

\begin{lemma}
The optimal solution to Eq.(\ref{EQ:click_CPA}) is given by:
\begin{eqnarray}
bid^* = \Big( arg_{bid} C &=&  
\frac{eCPM\_cost(bid)}{1000 \times pCTR \times pCVR} \Big), 
\end{eqnarray}
where Spend(bid) is within budget limit, 
\begin{eqnarray}
 Spend(bid) = Click(bid) \times  \frac{eCPM\_cost(bid)}{1000 \times pCTR}  \leq B
\end{eqnarray}
and 
\begin{eqnarray}
Conversion(bid) = Click(bid) \times pCVR
\end{eqnarray}
\end{lemma}
Note the above optimization requires pCTR/pCVR prediction\footnote{Predicted CTR and CVR depends on many features (not necessarily linear dependent on bid price), which will be illustrated in the following subsection.}, as well as bid landscape prediction (including eCPM\_cost and winrate).

{\bf pCTR/pCVR prediction} As is shown in many previous works, CTR/CVR prediction has been well studied in sponsored search advertising \cite{Graepel2010, McMahan2013, Zhu2010, Edizel2017
}, contextual advertising \cite{Ribeiro-Neto2005, Tagami2013}, display advertising \cite{Chapelle2014, Yan2014, Lee:2012:ECR:2339530.2339651}, and native advertising \cite{Li2015, He2014}. A common objective of these models is to accurately predict the probability of an ad being clicked once it is shown to end users. In our work, we incorporate history CTR/CVR observations with advertiser attributes (such as targeting location, device information, category and textual information, bid price, etc) for better prediction of CTR/CVRs.  We build a convolution neural network model~\cite{kongctr} for CTR prediction using robust feature learning\footnote{For cold-start ad without any history information, we adopt the exploration-exploit approach~\cite{Shah:2017:PES:3097983.3098041} to address the sparsity issue in conversion rate (e.g., 1\% conversions in 5\% of clicks) prediction.}. 

{\bf Bid landscape model prediction} Bid landscape model gives an overview of how 
many clicks one can win if the advertiser bids at a different price, which gives the picture of the advertiser auction behaviors in the marketplace. In this work, we use two components to 
describe the bid landscape model, including {\bf (i) win-rate model}  and {\bf (ii) eCPM\_cost model} which will be illustrated in detail in the next section.  These models are used for obtaining eCPM\_cost of Eq.(\ref{EQ:CPA}) and win\_rate of Eq.(\ref{EQ:click}), respectively.

 The history win rate, CPC cost, click through rate, conversion rate (CVR) on the advertisers and inventory become the observations used to determine future bids.  The learning process provides the estimation based on past data aggregated from other similar campaigns.  On the other hand it should be brief that some advertisers do not continue to spend money on impressions that don't perform in auctions, and therefore the learn phase needs to be long enough to collect adequate data.  
Chances are that we've collected a treasure chest of valuable data that can help predict more accurately over the incoming ads campaigns. In our practical setting, we always roll up our model and incorporate the new coming data everyday into the current infrastructures to turbocharge the learning process, which we call "roll up" learning process. 


{\bf Step 3: Enforcing the budget constraint}
If the optimization solutions obtained from Eq.(\ref{EQ:CPA}) satisfy the budget constraint, then we can simply return the recommended bid for advertisers and the optimization is done.  If not, the optimal solution to achieve the CPA goal does not exist.  There are two approaches that recommend advertisers to make adjustment: 

{\bf (i) Adjust the budget} Recommend the updated budget (usually increased budget $B'$) to advertisers, {\it i.e.,}
\begin{eqnarray}
B' = Spend(bid^*). 
\end{eqnarray}

{\bf (ii) Adjust the CPA goal} Recommend the updated CPA goal $C'$  the advertiser can achieve with the given budget, {\it i.e.,}
 \begin{eqnarray}
C' = CPA(bid'),  \;\;\;\;  \text{where} \;\;\;\;   bid' = (\arg_{bid} Spend(bid) = B ).
 \end{eqnarray}



\section{Bid Landscape Model Learning}


Bid Landscape performs like a ``bid simulator" to simulate how the auctions run in ad marketplace. In particular, bid landscape model looks at the specific auctions the advertiser participated in during the past to estimate the performance.  
In bid landscape models, win rate is a percentage metric which measures winning impressions over the total number of impressions.  For example, high win rate indicates low competition. The winning advertisers only need to pay the second highest bid price in a GSP auction. One existing model is survival model, which models the bid price during the investigation period from low price to high price as the patient's underlying survival period~\cite{DBLP:conf/kdd/ZhangZWX16}. However, it misses many unobserved prices in bid auctions. Further, it is counter-intuition that the win rate does not necessarily increase as the bid increases from ~\cite{DBLP:conf/kdd/ZhangZWX16}.  Lastly, the bidding price in auctions is not a sequential decision-making process as is assumed in survival model.

{\bf Overview }
\emph{The key idea is modeling eCPM bid ranges from comparing adjacent and non-adjacent positions and creating eCPM upper and lower bound for multiple context.  Then we translate these ranges to bins based on eCPM bid and winning opportunity probability density distributions.}  It consists of several steps:

{\bf Step 1:} Characterize the ranges of eCPM bid using the lowest possible eCPM bid ($eCPM\_dn$) and highest possible eCPM bid (denoted as $eCPM\_up$) which satisfies:
$eCPM\_dn  \leq eCPM \;\; bid \leq eCPM\_up$, and then derive the corresponding $eCPM\_cost$ based on GSP.  If the true eCPM bid is still lower than 
$eCPM\_dn$, one cannot win the ad display opportunity. On the other hand, if the true eCPM bid is higher than $eCPM\_up$, the advertisers will always win the first position but actually pay far less than the bid, which seldom happens in real-world. The lowest eCPM and highest eCPM give the range of possible eCPM that the advertiser can win the auction. (\S 4.1)



{\bf Step 2:} Generate Win-rate distributions over the eCPM bid price, {\it i.e.,} what is the percentage of winning impressions if one bids at the given eCPM bid price? The general idea is to generate the probability density function (\texttt{p.d.f}) and cumulative density function (\texttt{c.d.f}) for eCPM\_up, eCPM\_dn , respectively. Then the most likely win rate can be computed based on density distribution of winning opportunities. (\S 4.2)

{\bf Step 3:} Generate eCPM cost distributions over the eCPM bid price, {\it i.e.,} what is the most likely eCPM cost if one bids at the given eCPM bid price? The general idea is to generate the probability density function (\texttt{p.d.f}) and cumulative density function (\texttt{c.d.f}) for eCPM\_up, eCPM\_dn and eCPM cost lowest price and highest price, respectively. Then the most likely eCPM cost corresponding to current bid is computed based on density distribution of bid price. (\S 4.3)

{\bf Discussion}  We do not assume the bid price comes in sequence order, and therefore we are much closer to \emph{simulate} how different offers compete in an auctions independently.  Further the computed win-rate is also known as ``gross win rate'' by considering all the participated campaigns with the ad throttling rate that indicates how frequently or how recently a particular ad is allowed in an auction.  Moreover, we can ensure the higher bid gets more chance to win as in true auctions. 

\SetKwInOut{Parameter}{Parameters}
\KwIn{ Arrays of \{advertiser\}, \{context\}, \{position\}, \{pCTR\}, \{CPC bid\}, \{CPC cost\}, \{ranking\_score\}}
 \Parameter{max: maximum value of ecpm bid, {\it e.g., \$9.99}}
 \KwOut{Array of \{advertiser, context, position, ecpmup, ecpmdn, ecpm cost \}} 
\begin{algorithm}
$N = len(position)$ \;
\For{$i\ge 0$ \KwTo $n$ }{
\For{$j \ge 0$ \KwTo $n$ }{
\uIf{  $j==i$}{
    $ ecpmup = j>1 ? \frac{score[j-1]}{score[i]} \times bid[i] \times pctr[i] : max $\;
    $ ecpmdn = j<n ? \frac{score [j+1]}{score[i]} \times bid[i] \times pctr[i] : bid[i] \times pctr[i]$ \;
     }
  \uElseIf{ $j>i$}{
     $ecpmup =j>1? \frac{score[j]}{score[i]} \times  bid[i] \times pctr[i]: max  $ \;
     $ ecpmdn = j<n ? \frac{score[j+1]} {score[i]} \times bid[i] \times pctr[i] : bid[i] \times pctr[i]$  \;
  }
  \Else{
            $ecpmup = j>1 ?  \frac{score[j-1]}{score[i]} \times bid[i] \times pctr[i] : max$\;
            $ecpmdn = j<n ? \frac{score[j]}{score[i]} \times bid[i] \times pctr[i] : bid[i] \times pctr[i]$  \;
  }
   $ ecpmcost = cost[i] \times pctr[i]$  \;
    \If{$ecpmup>=ecpmdn$}
    { 
   \Return advertiser[i], context[i], position[j],  ecpmup,  ecpmdn,  ecpmcost
   }
}
}
\caption{Characterizing eCPM bid ranges and eCPM cost}
\end{algorithm}
\label{alg:ecpmbid}
\subsection{Characterizing the eCPM bid ranges}

One important observation is that ranking score provides the relative ranking of the neighboring positions finally used for ranking and pricing. Therefore, the key idea of our approach is to obtain the eCPM upper and lower bound by looking at the advertisers that appear at the neighboring positions.  More formally, the $j$-th position eCPM bid upper bound is given by $(j-1)$-th position bid, and eCPM bid lower bound is given by $(j+1)$-th position bid based on GSP because competitions occur at the neighboring positions, {\it i.e.,}
{
\begin{eqnarray}
eCPM\_up[j]  = \frac{ranking\;\;score[j-1]} {ranking\;\; score[j]} \times bid[j] \times pCTR[j], \nonumber \\
eCPM\_dn[j]  =  \frac{ranking\;\; score[j+1]} {ranking\;\; score[j]} \times bid[j] \times pCTR[j],
\end{eqnarray}
}
where $score[j]$ denotes j-th position ranking score, $pCTR[j]$ denotes the j-th position predicted CTR, $bid[j]$ denotes j-th position bid.  Notice that in practical auctions,   even if positions $i$ and $j$ are not neighbors, 
we can still get useful information regarding the upper and lower bound for pricing. Actually, for any two positions $i, j$, we have:  
{
\begin{eqnarray}
 \forall j>i , \;\;\;\; ecpmup[j] =  \frac{ranking\;\;score[j]} {ranking\;\; score[i]} \times bid[i] \times pCTR[i], \\
 \forall j<i, \;\;\;\;  ecpmdn[j] = \frac{ranking\;\;score[j]} {ranking\;\; score[i]} \times bid[i] \times pCTR[i].
\end{eqnarray}
}

This, in fact, provides opportunity to explore all the available logs to infer the eCPM bid ranges and corresponding cost by traversing all pairwise advertisers with adjacent or non-adjacent positions. Alg.~\ref{alg:ecpmbid} presents the details of characterizing the eCPM bid ranges.  The time complexity of this algorithm is $\mathcal{O}(n^2)$, where $n$ is the number of possible positions. For example, given a simple auction with three advertisers, we have:

{

\texttt{array of Advertisers:  \{9192982670, 9620472854, 9575604786\}}

\texttt{array of context:  \{1\_mobile,  1\_desktop,  1\_mobile\},  }

\texttt{array of Positions: \{1,2,3\}, }

\texttt{array of ranking scores:  \{3.117E-4, 2.387E-4, 2.312E-4\}},

\texttt{array of Bid: \{5.0, 1.581, 0.5\}, }

\texttt{array of Cost: \{0.31,0.5, 0.45\},  }

\texttt{array of pCTR:  \{0.002588, 8.0119E-4, 5.7167E-4\}}
}

\begin{figure}[t]
	\centering
	\includegraphics[height=1.4in,width=0.21\textwidth]{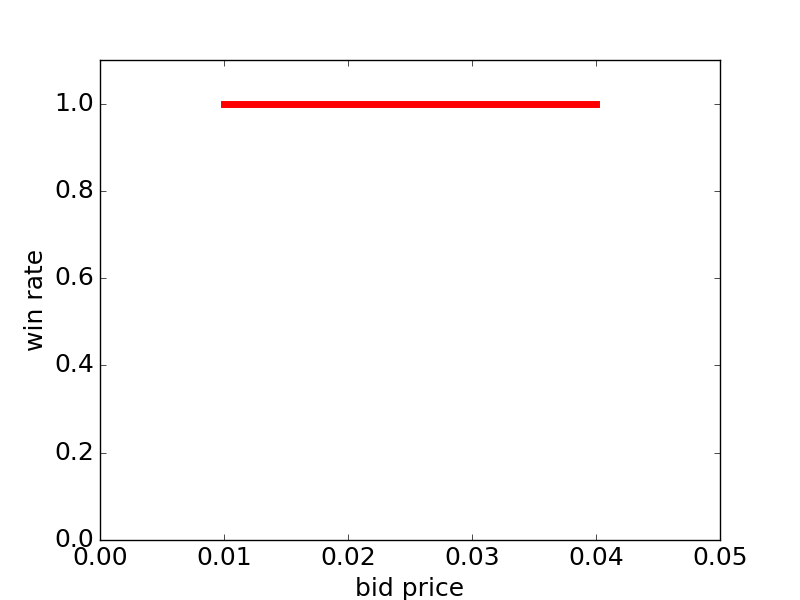}
	\includegraphics[height=1.4in,width=0.21\textwidth]{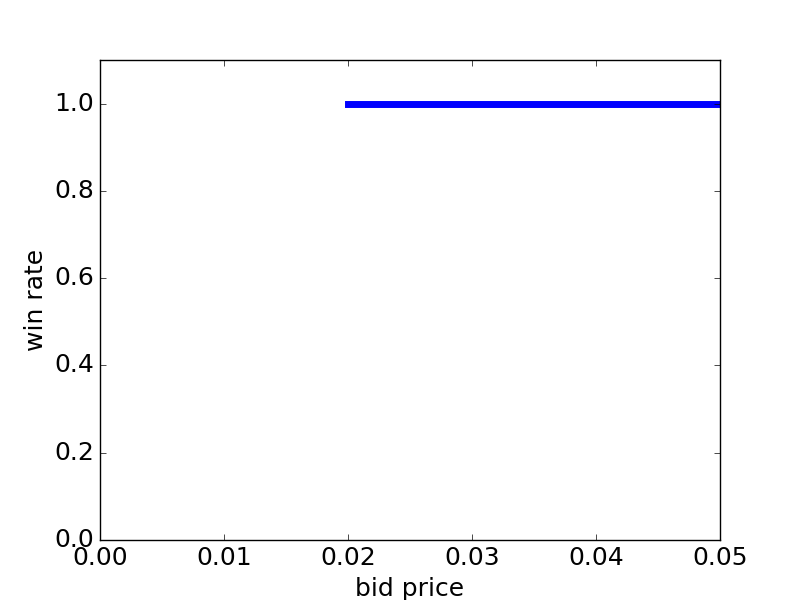}
	\includegraphics[height=1.4in,width=0.21\textwidth]{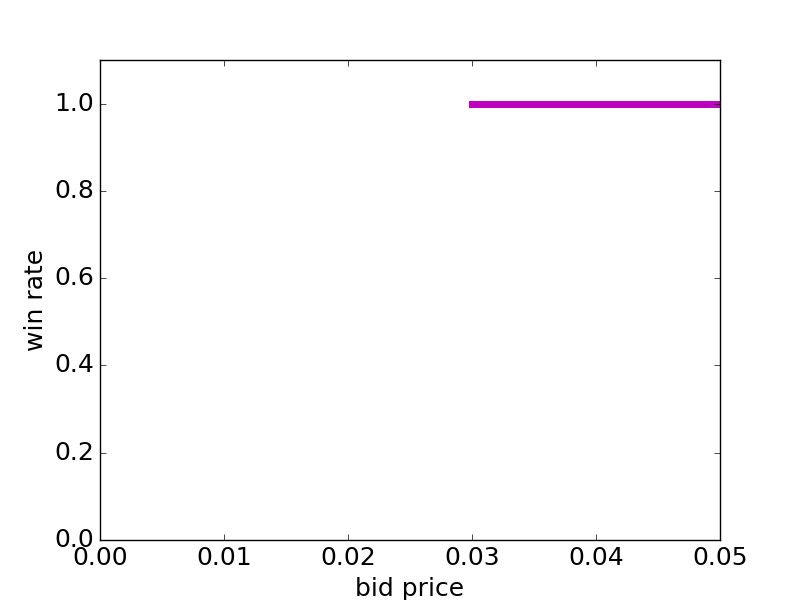}
	\includegraphics[height=1.4in,width=0.21\textwidth]{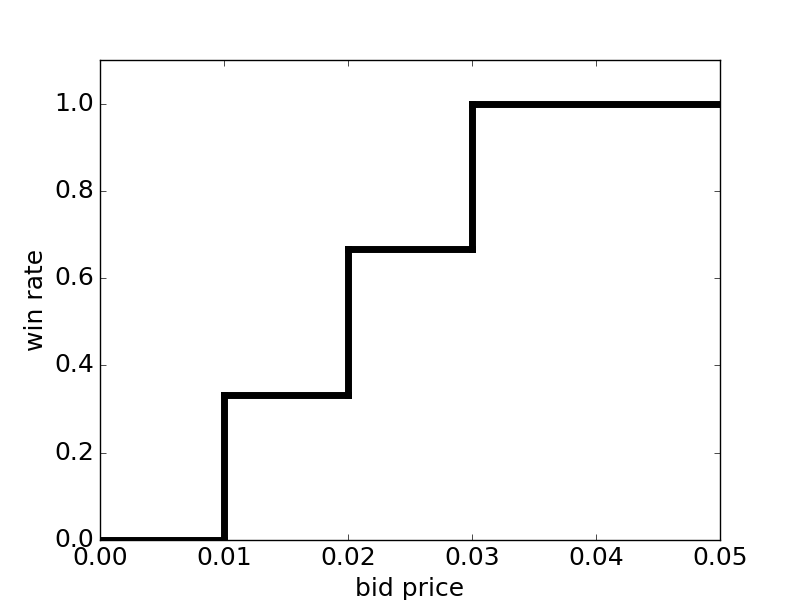}
		\caption{\small{The first three figures demonstrate the three observations with \{eCPM up, eCPM down, eCPM cost\}, and the bottom-right figure is the generated win-rate distribution over eCPM bid price. }
 }
	\label{fig:winrate}
\end{figure}

The output will be the complete arrays of \{advertiser, context, position, ecpmup, ecpmdn, ecpm cost \} by considering different position combinations used for building bid-landscape, {\it i.e.,}

{

\texttt{position = 1}

\texttt{ \{9192982670, 1\_mobile, 1, 9.99,0.0099,0.00080\} },  

\texttt{ \{9620472854, 1\_desktop, 1,  9.99,0.001654,0.0004\} }, 

\texttt{ \{9575604786, 1\_mobile, 1, 9.99,0.0003852,0.00026\}}; 

\texttt{position = 2}

\texttt{ \{9192982670, 1\_mobile, 2, 0.0099,0.00960,0.000802\}},

\texttt{ \{9620472854, 1\_desktop, 2, 0.00165,0.00122,0.000400\}}, 

\texttt{ \{9575604786, 1\_mobile, 2, 0.000385,0.000294,0.000257\}}; 

\texttt{position = 3}

\texttt{ \{9575604786,  1\_mobile, 3,0.000294,0.000285,0.000257\}}. 
}

\subsection{Win-rate distribution over eCPM bid}

Given the learned eCPM bid lower and upper bound, eCPM cost, we can generate win-rate distributions over eCPM bid price. The general idea is to compute the ratio of winning auctions during each available eCPM bid price.  First we bucket the eCPM bid highest and lowest arrays into evenly distributed bins, and get the corresponding {\bf \emph{index}} for each bid price. For example, if the bin size = \$0.01,  bid price \$0.05 falls into the fifth bin due to $\frac{0.05}{0.01}=5$.  Given the distribution of bid prices,  it is easy to get the probability density distributions  (\texttt{p.d.f} ) of eCPM upper and lower bounds, followed by the corresponding cumulative density distribution  (\texttt{c.d.f} ) of them.  Finally win-rate is given by:

{
\begin{lemma}
Win rate at different bid price with $index = \frac{bid}{bin\_size}$ is given by: 
\begin{eqnarray}
winrate(index) = \frac{ \text{ c.d.f of ecpm\_dn (index) - c.d.f of ecpm\_up (index) } }  {n}, \nonumber
\end{eqnarray}
where $n$ is the number of observations and c.d.f of ecmp\_up and ecpm\_dn is computed based on CPM bid distributions\footnote{\small {Here c.d.f and p.d.f are similar to histograms, which are finally normalized by division of $n$ observations.}}.
\end{lemma} 
}

The rationality behind this algorithm is that \emph{only} the ad campaigns with bidding price falling between the lower and upper bound of eCPM bid can win the auction, and therefore the number of campaigns falling into this range divided by the total number of observed campaigns gives the correct ratio of win-rate. 
Alg.~\ref{Alg:winrate} presents the detailed  learning algorithm.  For example, given three observations of the example shown in \S 4.1, we have 
arrays of \{eCPM up, eCPM down, eCPM cost\} shown as follows with $bin\_size=0.01$:  

{
\texttt{Observation 1: \{0.04, 0.01, 0.008\}  },

\texttt{Observation 2:  \{0.05, 0.02, 0.015\} },

\texttt{Observation 3: \{0.05, 0.03, 0.02\}},
}

The win-rate model training process is presented in Table~\ref{tbl:winrate} and the learned win-rate model is plotted in the bottom-right figure of Fig.\ref{fig:winrate}.
{
\SetKwInOut{Parameter}{Parameters}
\KwIn{Array of \{eCPM up, eCPM down, eCPM cost\}}
 \Parameter{n: the number of observations in array;  bin\_size: the interval of eCPM bid, say=0.01; max: the largest index of bid given the bin\_size}
 \KwOut{Win-rate array at each observation $i$} 
\begin{algorithm}
\For{$i\ge 0$ \KwTo $n$ }{
indmin = ecpm dn[i]/bin\_size; \\
indmax = ecpm up[i]/bin\_size;\\
  \If{$indmax > max $}
   { $max = indmax$
 }
   \If{$indmin > 0 \;\; and\;\;  indmax>0$}
   { $\texttt{p.d.f}\;\; dn[indmin] ++ $ \;
    $\texttt{p.d.f} \;\; up [indmax] ++ $ \; }
   $i := i + 1$
 }
\For{$j\ge 0$ \KwTo $max$ }{ 
    $\texttt{c.d.f} \;\; dn[j] += \texttt{p.d.f} \;\;dn [j]$ \;
     $\texttt{c.d.f} \;\; up [j] += \texttt{p.d.f} \;\; up [j]$ \;
     $ winrate[j] :=  ( \texttt{c.d.f} \;\; dn [j] - \texttt{c.d.f} \;\; up [j])/n $ \;
     $j := j + bin\_size$
    }
   \Return Win-rate array
\caption{Win-rate distribution learning}
\label{Alg:winrate}
\end{algorithm}
}

{
\begin{center}
\begin{table}
  \begin{tabular}{ c |  c | c | c | c | c  }
    \hline \hline
    eCPM bid  & 0.01 & 0.02 & 0.03 & 0.04 & 0.05 \\ 
    \hline 
    Index & 1 & 2 & 3 & 4 & 5 \\
    \hline
    \texttt{p.d.f} of ecpm down & 1 & 1 & 1 & 0 & 0 \\
     \hline
  \texttt{p.d.f} of ecpm up & 0 & 0 & 0 & 1 & 2 \\
    \hline
  \texttt{c.d.f}  of ecpm down & 1 & 2 & 3 & 3 & 3 \\
    \hline
  \texttt{c.d.f}  of ecpm up & 0 & 0 & 0 & 1 & 3 \\
   \hline \hline
  \end{tabular}
  \caption{Win rate model training process.  p.d.f and c.d.f for all ecmp up and dn prices are finally normalized by $n$ observations.}
    \label{tbl:winrate}
\end{table}
\end{center}
}
\begin{figure}[t]
	\centering
	\includegraphics[height=1.4in,width=0.21\textwidth]{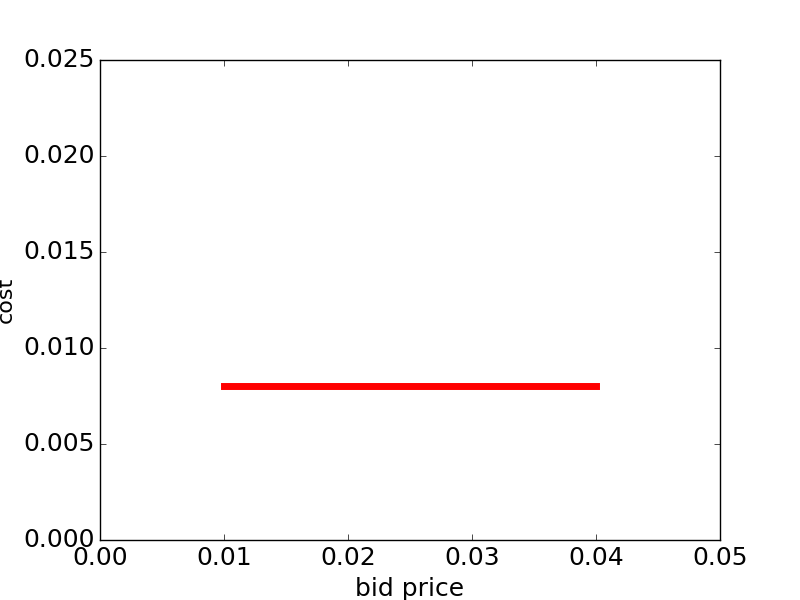}
	\includegraphics[height=1.4in,width=0.21\textwidth]{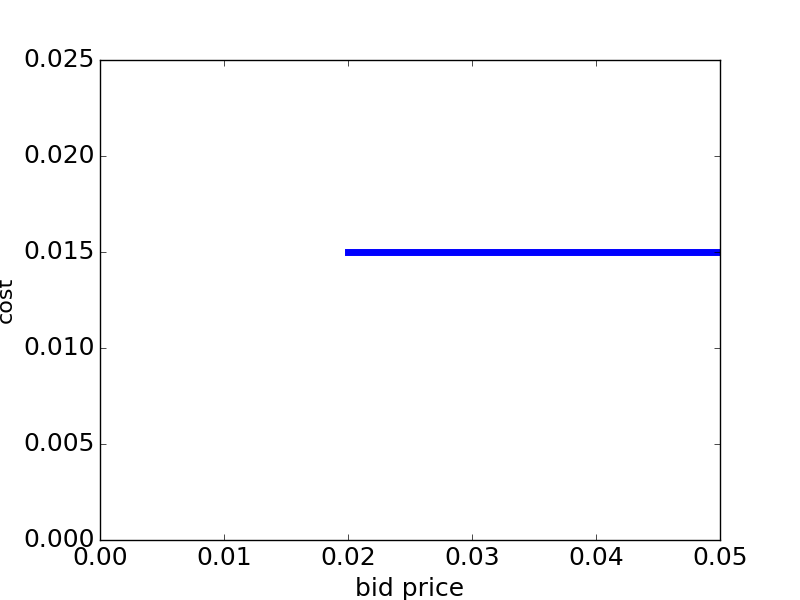}
	\includegraphics[height=1.4in,width=0.21\textwidth]{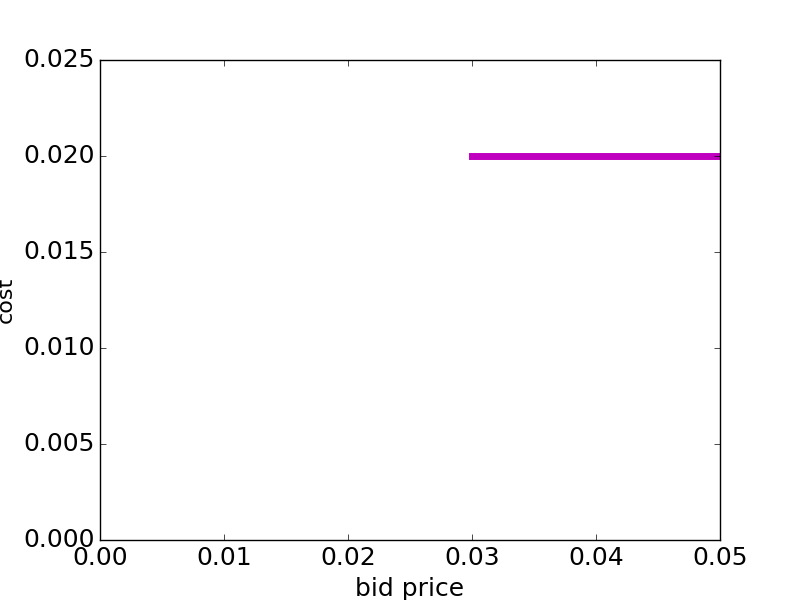}
	\includegraphics[height=1.4in,width=0.21\textwidth]{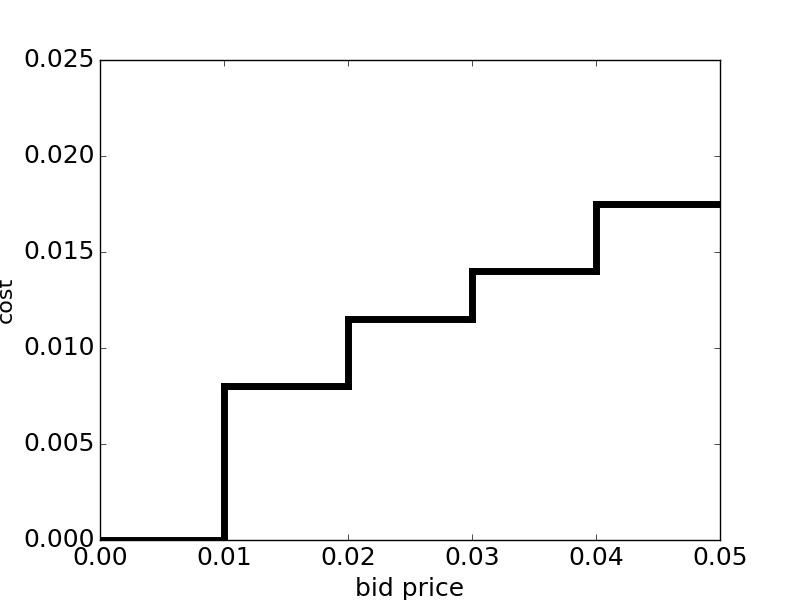}
		\caption{\small{The first three figures demonstrate the three observations with \{eCPM up, eCPM down, eCPM cost\}, and the bottom-right figure is the generated eCPM cost distribution over eCPM bid price. }  }
	\label{fig:cost}
\end{figure}
{
\begin{center}
\begin{table}
  \begin{tabular}{ c |  c | c | c | c | c }
    \hline
    eCPM bid price & 0.01 & 0.02 & 0.03 & 0.04 & 0.05 \\ 
    \hline  \hline
    Index & 1 & 2 & 3 & 4 & 5 \\
    \hline
\texttt{p.d.f} of ecpm down & 1 & 1 & 1 & 0 & 0 \\
     \hline
\texttt{p.d.f} of ecpm up & 0 & 0 & 0 & 1 & 2 \\
    \hline
\texttt{p.d.f} ecpm cost down & 0.008 & 0.015 & 0.02 & 0 & 0 \\
    \hline
\texttt{p.d.f} of ecpm cost up  & 0 & 0 & 0 & 0.008 & 0.035 \\
    \hline
   \texttt{c.d.f} of ecpm down & 1 & 2 & 3 & 3 & 3 \\
    \hline
   \texttt{c.d.f}  of ecpm up & 0 & 0 & 0 & 1 & 3 \\
    \hline 
   \texttt{c.d.f}   of cost down & 0.008 & 0.023 & 0.043& 0.043 & 0.043 \\
    \hline
   \texttt{c.d.f}  of cost up &  0 & 0 & 0 & 0.008 & 0.043 \\ 
    \hline 
   \hline
  \end{tabular}
 \caption{\small eCPM cost distribution learning algorithm}
\label{tbl:cost} 
  \end{table}
   \end{center}
}

\subsection{eCPM cost distribution over  eCPM bid}
Given the learned eCPM bid lower and upper bound, eCPM cost, we generate eCPM cost distributions over eCPM bid price. The general idea is to compute the \emph{average eCPM cost} for each available eCPM bid price.  First we estimate the probability density distributions  (\texttt{p.d.f} ) of eCPM upper bound and lower bound, as well as eCPM cost upper bound and lower bound, followed by the computation of the corresponding cumulative density distribution  (\texttt{c.d.f} ) of them. Compared to win-rate model, we need to  track \emph{extra eCPM cost upper and lower bound } in order for accurate estimation of the true cost. Finally, eCPM cost  is given by:
{
\begin{lemma}
eCPM cost at different bid price with $index = \frac{bid}{bin\_size}$ is given by: 
\begin{eqnarray}
eCPM\_cost(index) = \frac{ {c.d.f cost dn(index) - c.d.f cost up(index) } }  { {c.d.f ecpm dn(index) -  c.d.f  ecpm up(index)}} \nonumber
\end{eqnarray}
\end{lemma}
}
Similar to win-rate,  \emph{only} ad campaigns with cost price falling between the lower and upper bound of eCPM bid can win the auction, and therefore the number of campaigns falling into this range is divided by the total number of observed campaigns whose ecpm bid prices fall between ecpm bid lower and upper bound. 
To summarize, Alg.~\ref{alg:cost} presents the algorithm.  
As an illustrative example, with three observations shown in \S 4.2, the eCPM cost  training process is presented in Table~\ref{tbl:cost} and the learned cost model is plotted in the bottom-right figure of Fig.\ref{fig:cost}.


{
\SetKwInOut{Parameter}{Parameters}
 \KwIn{Array of \{eCPM up, eCPM down, eCPM cost\}}
 \Parameter{n: the number of observations in array;  bin\_size: the interval of eCPM bid, \\
 say=0.01; max: the largest index of bid given the bin\_size}
 \KwOut{eCPM cost array at each observation $i$} 
\begin{algorithm}
\For{$i\ge 0$ \KwTo $n$ }{
indmin = ecpm down[i]/bin\_size; \\
 indmax = ecpm up[i]/bin\_size;  \\
cost = ecpm cost[i]; \\
  \If{$indmax > max $}
   { $max = indmax$
 }
   \If{$indmin > 0 \;\; and\;\; indmax>0$}
   { $ \texttt{p.d.f}\;\; dn[indmin] ++ , \;\;  \texttt{p.d.f} \;\;up [indmax] ++ $ \; 
    $\texttt{p.d.f}\;\; cost\;\; dn[indmin] +=cost $ \;
    $\texttt{p.d.f}\;\;cost \;\;up [indmax] +=cost $\;
    }
    $i := i + 1$
 }
\For{$i\ge 0$ \KwTo $max$ }{ 
    $ \texttt{c.d.f} \;\;dn[i] +=\texttt{p.d.f.} \;\; dn [i]$\;
    $  \texttt{c.d.f} \;\; up [i] += \texttt{p.d.f} \;\;up [i]$ \;
     $\texttt{c.d.f}\;\; cost dn[i] +=\texttt{ p.d.f} \;\;cost\;\; dn [i]$ \;
     $\texttt{c.d.f} \;\; cost up [i] += \texttt{p.d.f}\;\; cost\;\; up [i]$ \;
     $ecpm cost[i] :=  ( \texttt{c.d.f}\;\; cost\;\; dn [i] - \texttt{c.d.f} \;\; cost \;\;up [i])/  (\texttt{ c.d.f.}\;\;  dn [i] -\texttt{ c.d.f.}\;\; up [i])$ 
     $j := j + bin\_size$
    }
    \Return eCPM cost array
\caption{eCPM cost distribution learning algorithm}
\label{alg:cost}
\end{algorithm}


\section{Performance evaluation}

{\bf Advertising campaign dataset} We collect advertising campaign log data on a major advertising platform in
U.S after removing inactive campaigns. We take a sample of 15\% of all bidding records in one-week, totally 135,409 native ad campaign 
profiles in offline evaluation.   All these campaigns focus on clicks for displaying advertisement that well matches the form and function of the publisher platform upon which it appears.
The campaigns use cost-per-click (CPC) bidding,  and the advertiser needs to pay only when someone actually clicks on the ad after visiting publisher site. In  auctions the advertising campaigns win the participated inventory based on GSP.  
The advertising campaign conversions are associated with conversion rules, which lead to 
user-action such as as 
purchase, sign-up, registration, or view a key page. Since different
conversion actions have uneven difficulty, advertisers usually set
the customized conversion rules in a campaign to track
the conversion actions they care.  In the dataset there are totally 277
conversion rules, and 244 targeting definitions which consist of
53 geographical regions (states and areas) on seven  device platforms. 
 We show the distribution of bid price,
 CTR of advertisers, and top 10 conversion rules in Fig.\ref{fig:insight}.

\subsection{Offline evaluation}


\subsubsection{CPA goal optimization performance}
Given the fact that we only know the ground-truth for the current bid, {\it i.e.,} the true number of conversions and the true spend for current campaign at current bid price,   we therefore design experiments to compare the differences between the true CPA (denoted as CPA) and the forecasted CPA (denoted as $\hat{CPA}$) at current bid.  More formally, we use the mean absolute percentage error (MAPE) and root mean square percentage error (RMSPE)  as the measurement, {\it i.e.,}  
{
\begin{eqnarray}
\text{MAPE} =  \sum_{i=1}^n \frac { |{y}_i-  \hat{y_i}| }{y_i},    \;\;\; 
\text{RMSPE} = \sqrt{\frac{1}{n}\sum_{i=1}^n{  ( \frac{\hat{y_i}-y_i}{y_i}})^2},
\label{EQ:ape}
\end{eqnarray}
}
where $y_i$ is the CPA for ads campaign $i$ with current bid $b_i$, and $\hat{y}_i$ is the forecasted CPA for campaign $i$ at current bid $b_i$.
The smaller values of MAPE, RMSPE indicate better performance.  We also compare the following baseline methods. 

$\bullet$ {\bf Nearest Neighbor Search (NNS).} 
For the observed bid $b$,  if $b'_i$ is the nearest neighbor of bid $b_i$ in history, then we retrieve the corresponding CPA for bid $b'_i$ as the predicted CPA.  
 
$\bullet$  {\bf Linear interpolation (LI)} 
For the observed true bid $b$,  if $b_i \leq b \leq b_j$, where $b_i$ and $b_j$ are the nearest bid prices in history observations.
Then the predicted $\hat{CPA}$ at bid $b$ is the linear interpolation 
\begin{eqnarray}
\hat{CPA}= C_i + (b- b_i)\frac{C_j - C_i}{b_j - b_i}, 
\end{eqnarray}
where $C_i$ is the observed CPA at bid $b_i$ and $C_j$ is the observed CPA at bid $b_j$.
based on $b_i$ and $b_j$.  

{
\begin{figure*}
	\centering
	\begin{subfigure}[b]{0.30\textwidth}
	\includegraphics[height=1.4in, width=\linewidth]{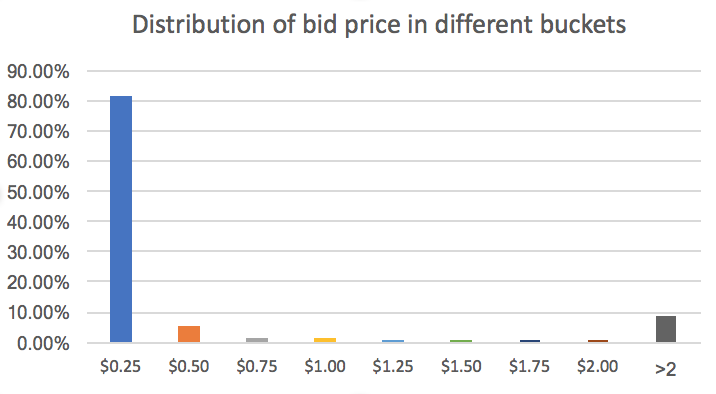}
	\caption{Bid distribution}
		\label{fig:bid}
	\end{subfigure}
	\begin{subfigure}[b]{0.30\textwidth}
	\centering
	\includegraphics[height=1.4in, width=\linewidth]{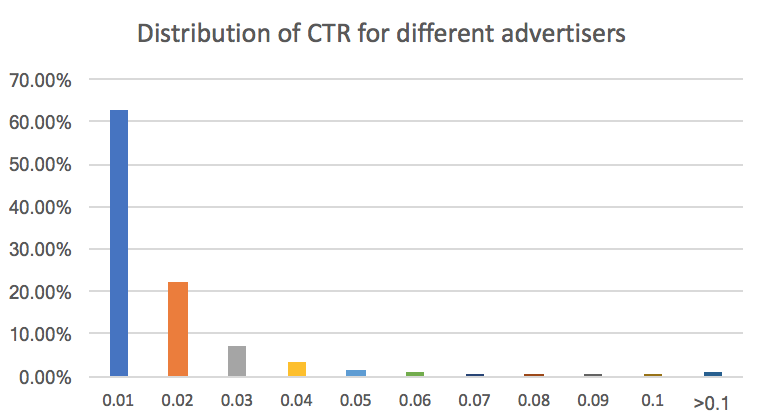}
		\caption{CTR distribution}
	\label{fig:ctr}
	\end{subfigure}
	\begin{subfigure}[b]{0.30\textwidth}		
	\centering
	\includegraphics[height=1.4in, width=\linewidth]{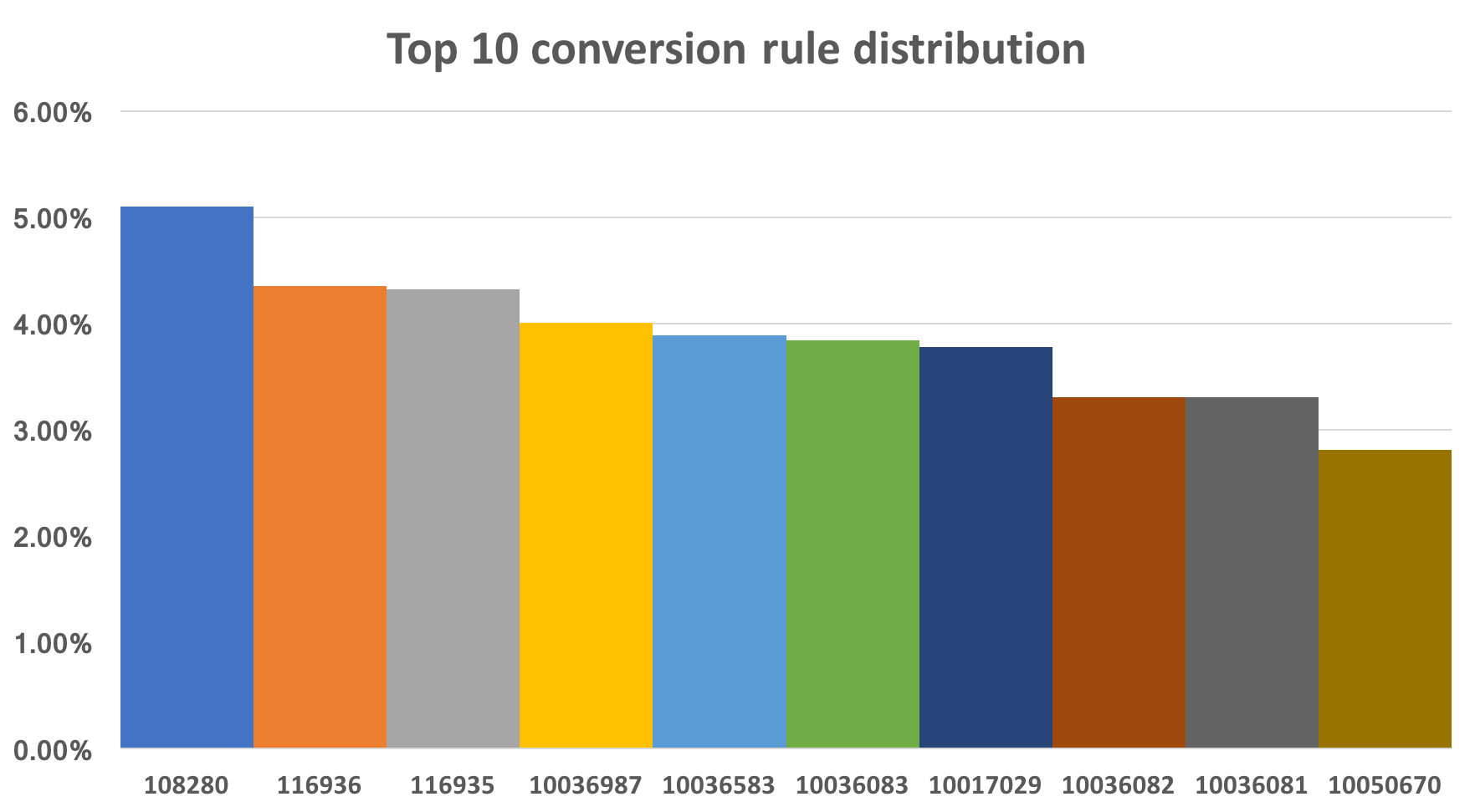}
	\caption{Conversion rule (top 10) distribution}
	\label{fig:supply}
	\end{subfigure}		
	\caption{\small (a) distribution of bid price; (b) distribution of CTR for advertisers; (c) distribution of top 10 conversion rules. }
	\label{fig:insight}
\end{figure*}
}

$\bullet$ {\bf Gradient Boost Decision Tree method (GBDT)} The idea is to build a regression model for CPA prediction given the features including the bid price, advertiser ID and other advertiser features (such as supply group, advertiser category, etc)~\cite{DBLP:conf/kdd/CuiZLM11}.  When an advertiser bids in an auction, GBDT tells what is the forecasted CPA. GBDT~\cite{Mason00boostingalgorithms} model essentially is in the form of an ensemble of weak prediction models such as decision trees.  In experiment, we set tree number=25, number of leaf node=40, learning rate=0.75 and sampling rate=0.75 as a regression task for CPA prediction.

{
\begin{table}
\caption{\small Performance comparisons of MAPE, RMSPE between CPA and predicted $\hat{CPA}$ using (1) proposed method; (2) baseline 1: nearest neighbor search (NNS); (3) baseline 2: linear interpolation (LI); (4)  baseline 3: GBDT}
\label{tbl:CPA_ape}
\begin{center}
\begin{tabular}{ c|c|c}
\hline
\hline
Method & MAPE & RMSPE \\
\hline
Our method & 10.28\%  & 15.68\% \\
Nearest Neighbor Search (NNS)  &  23.86\% &  50.23\% \\
Linear Interpolation (LI)  &  18.65\%  & 32.65\% \\
GBDT   &  20.07\%  & 31.09\% \\
\hline
\hline
\end{tabular}
\end{center}
\end{table}
}

{\bf Insight} Table. ~\ref{tbl:CPA_ape} shows the performance evaluation on the given ad campaign dataset from a major web portal company.  Given that it is a challenging task to predict CPA for diversified, fluctuated ad campaigns, the result obtained using proposed model is competitive and reasonable, because it leverages all history information completely by smartly setting the bid price. The complete history information makes the model more robust by smoothing the outliers and noises in observations. Moreover, the model provides the principled way to achieve the optimization goal. 

One important observation is the forecast result from GBDT is not accurate, due to the fact that the available advertiser-level features are far smaller than user-level features. More importantly, the non-decreasing non-linear relations between CPA and bid price is a bit hard to capture by using the tree structure in the traditional classification tasks.  NNS approach cannot model the CPA well due to the simple assumption of relations between CPA and bid price. Moreover, our algorithm is linear with the number of campaigns in participation in the same auction, therefore it can scale very well with large scale ad campaigns running in internet companies. 
Fig.\ref{fig:cpa-click-spend} shows several examples of ad campaigns with forecast CPA, total money of spend and number of clicks at different bid prices. 



\begin{figure*}
	\centering
	\begin{subfigure}[b] {0.33\textwidth}
	\includegraphics[height=2in, width=\linewidth]{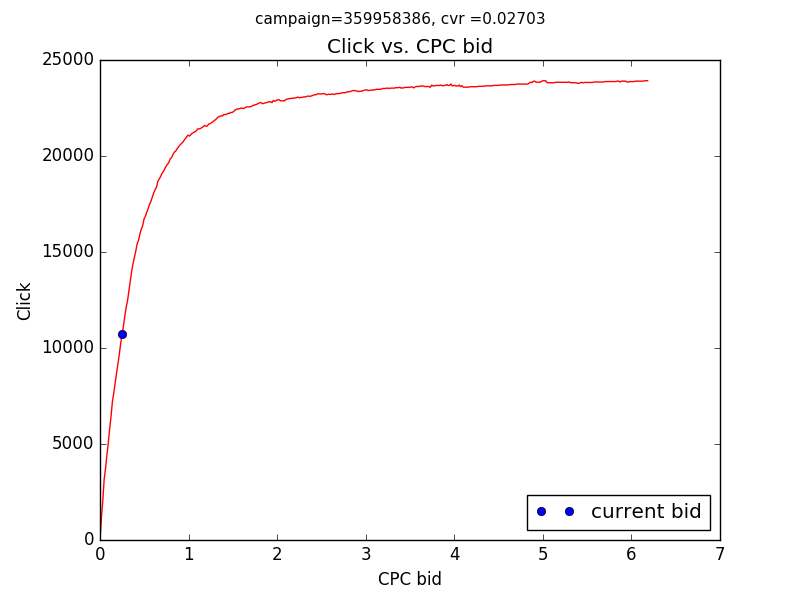}
	\caption{Ads 1: click vs. bid}
		\label{fig:click-1}
	\end{subfigure}
	\begin{subfigure}[b]{0.33\textwidth}
	\centering
	\includegraphics[height=2in, width=\linewidth]{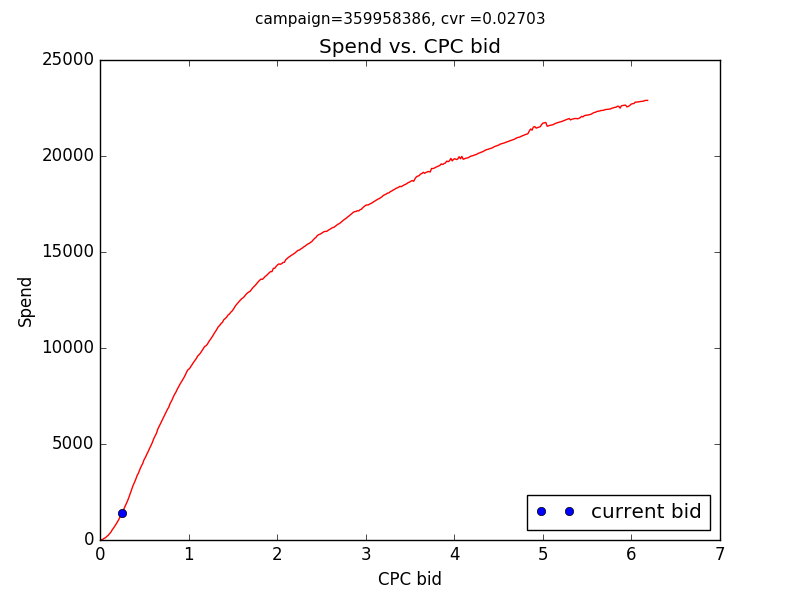}
		\caption{Ads 1: spend vs. bid}
	\label{fig:spend-1}
	\end{subfigure}
	\begin{subfigure}[b]{0.33\textwidth}		
	\centering
	\includegraphics[height=2in, width=\linewidth]{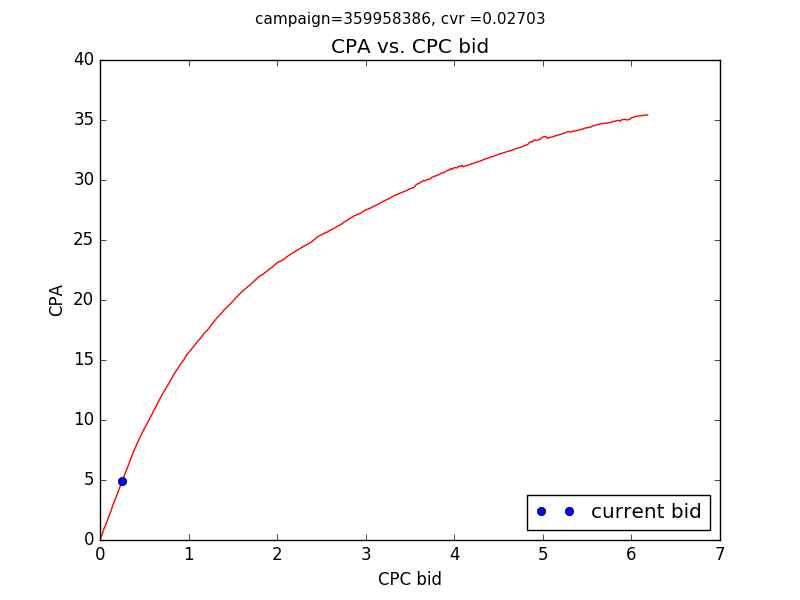}
	\caption{Ads 1: cpa vs. bid}
	\label{fig:cpa-1}
	\end{subfigure}	
	\begin{subfigure}[b]{0.33\textwidth}
	\includegraphics[height=2in, width=\linewidth]{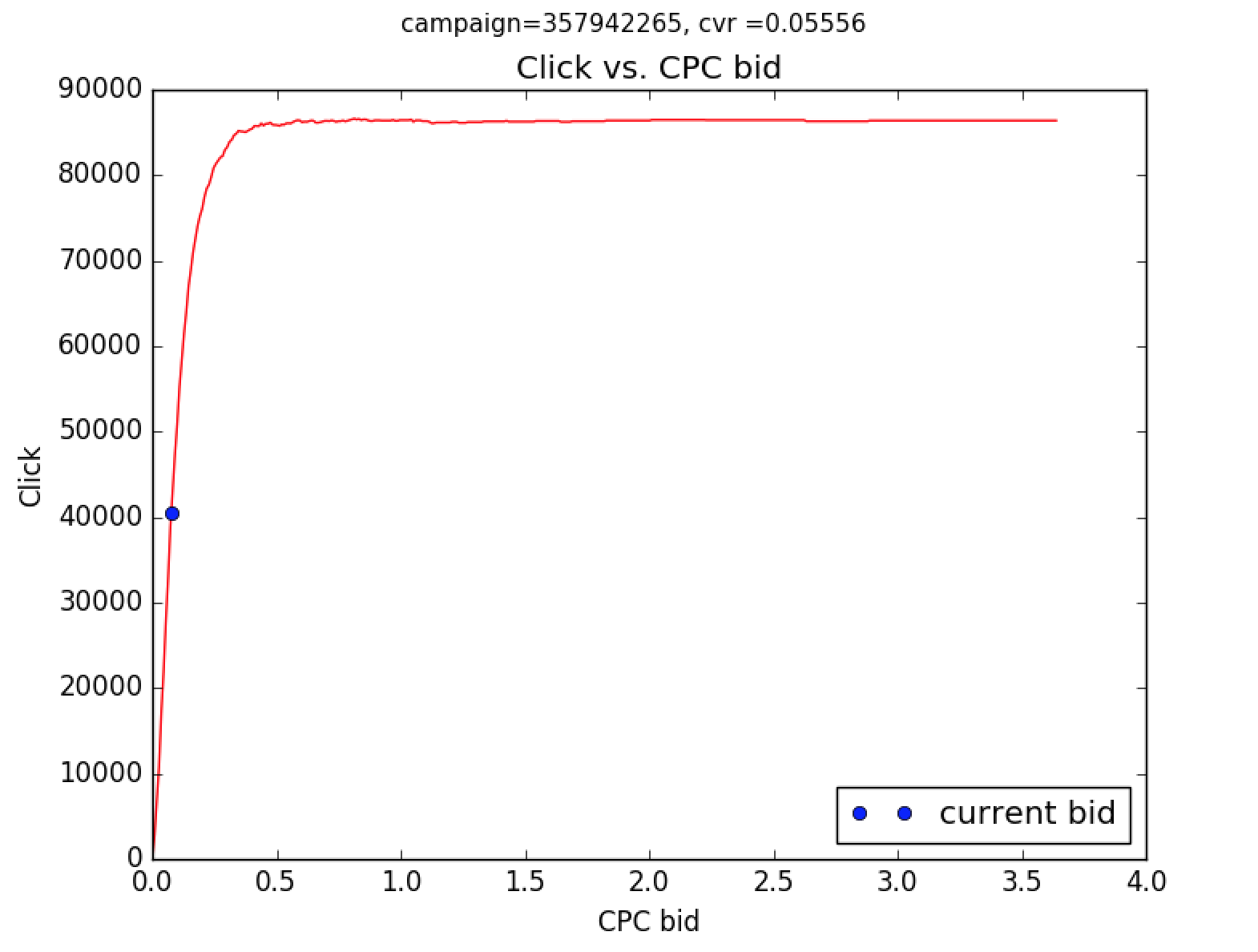}
	\caption{Ads 2: click vs. bid}
		\label{fig:click-2}
	\end{subfigure}
	\begin{subfigure}[b]{0.33\textwidth}
	\centering
	\includegraphics[height=2in, width=\linewidth]{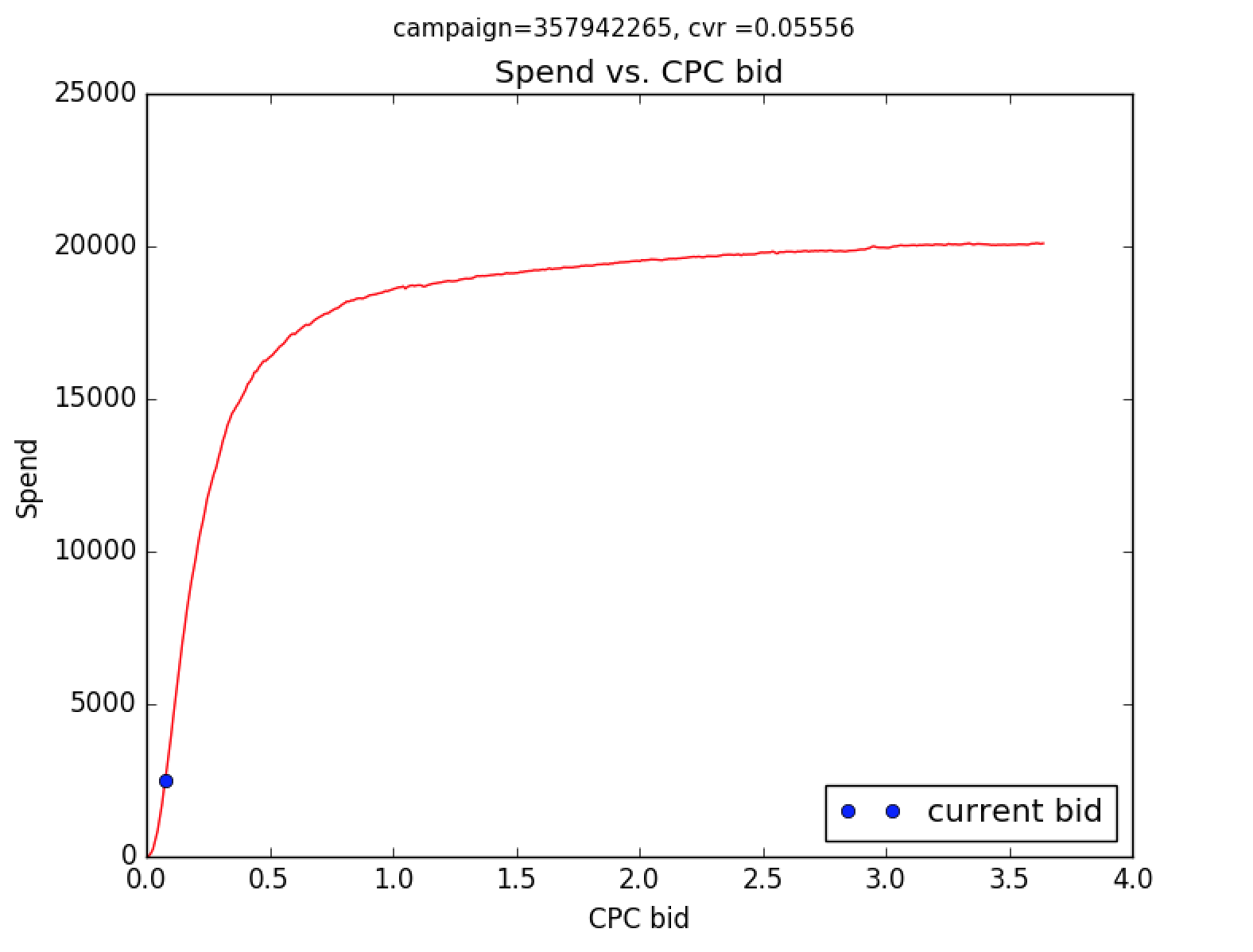}
		\caption{Ads 2: spend vs. bid}
	\label{fig:spend-2}
	\end{subfigure}
	\begin{subfigure}[b]{0.33\textwidth}		
	\centering
	\includegraphics[height=2in, width=\linewidth]{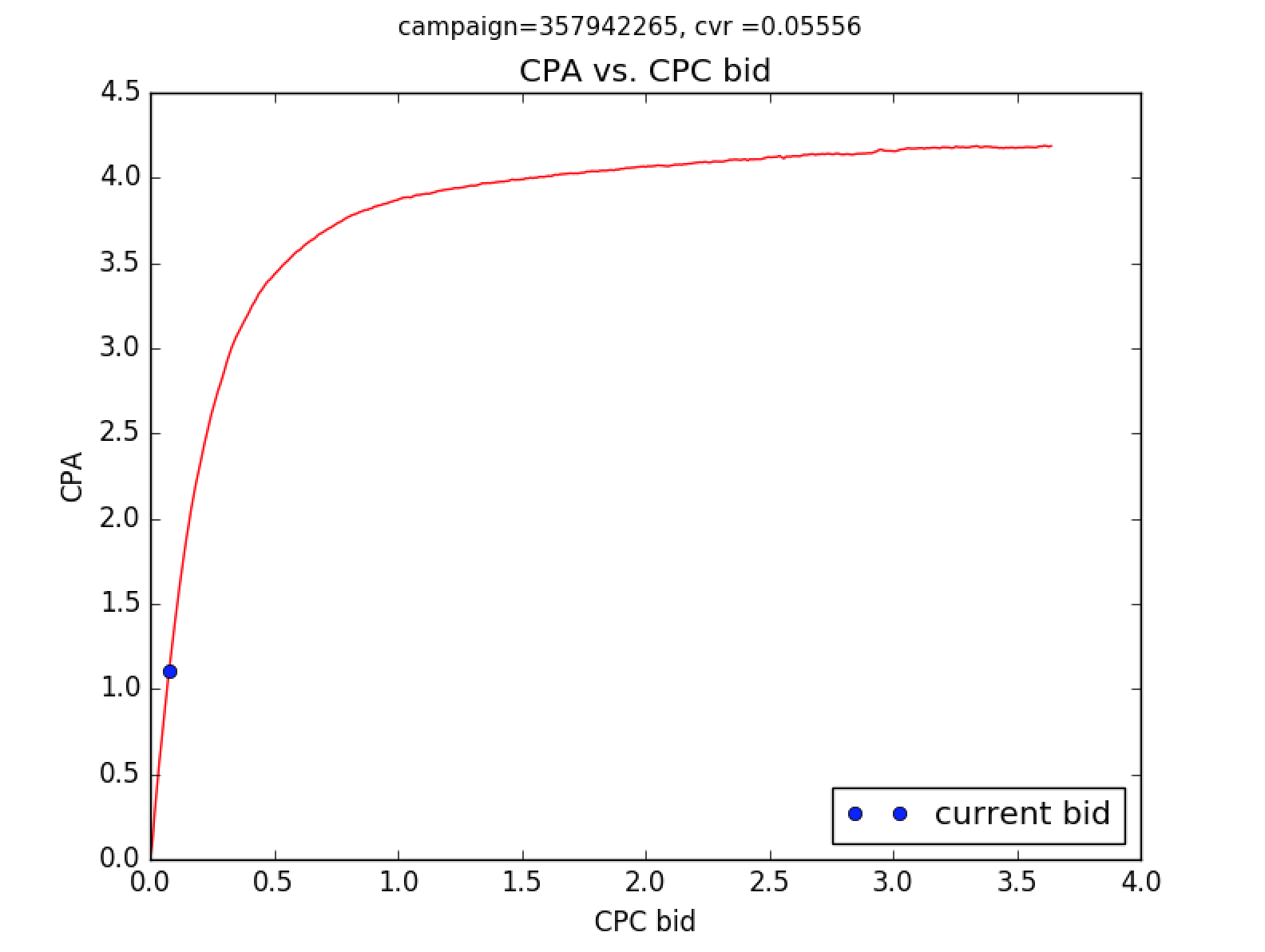}
	\caption{Ads 2: cpa vs. bid}
	\label{fig:cpa-2}
	\end{subfigure}	
		\begin{subfigure}[b]{0.33\textwidth}
	\includegraphics[height=2in, width=\linewidth]{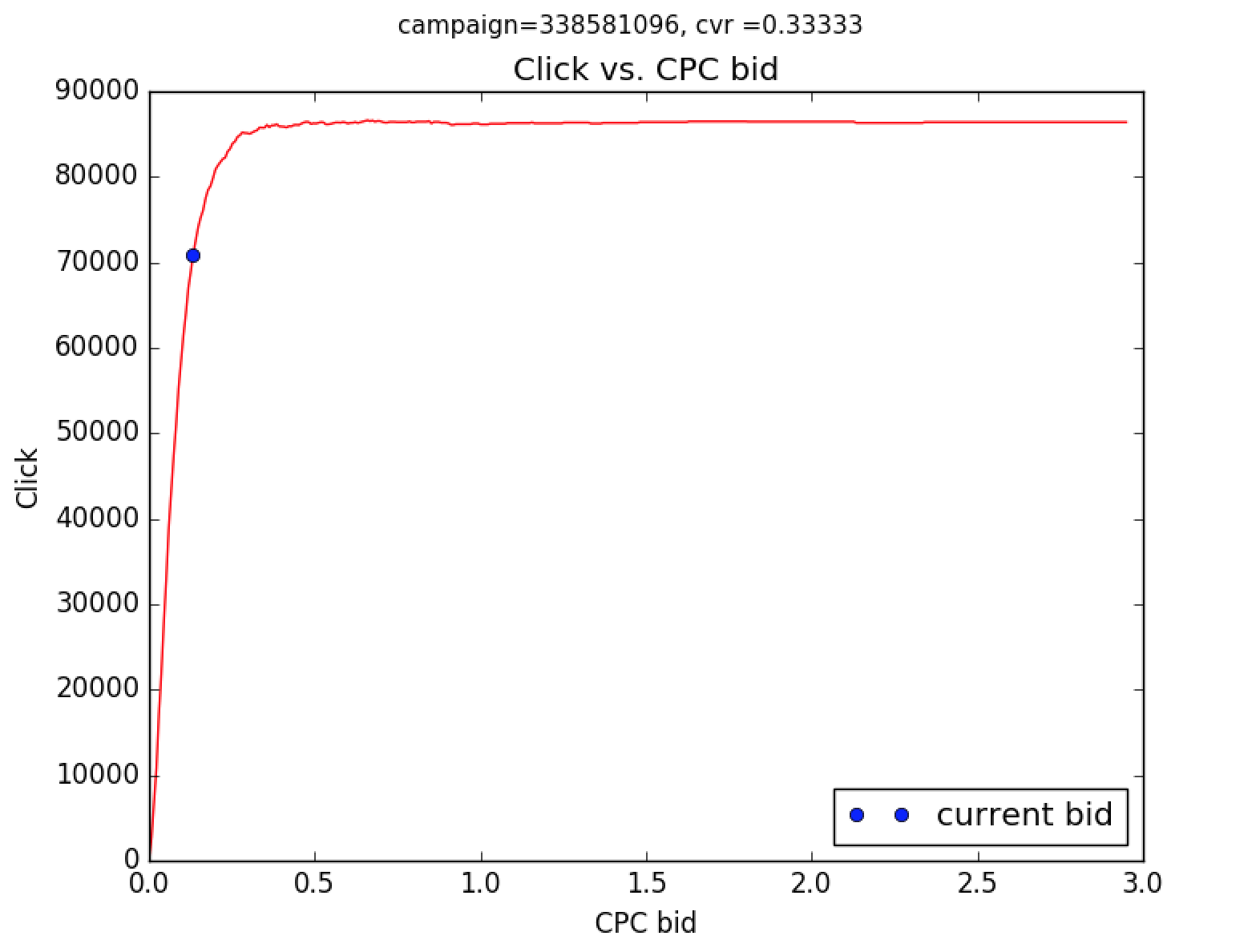}
	\caption{Ads 3: click vs. bid}
		\label{fig:click-3}
	\end{subfigure}
	\begin{subfigure}[b]{0.33\textwidth}
	\centering
	\includegraphics[height=2in, width=\linewidth]{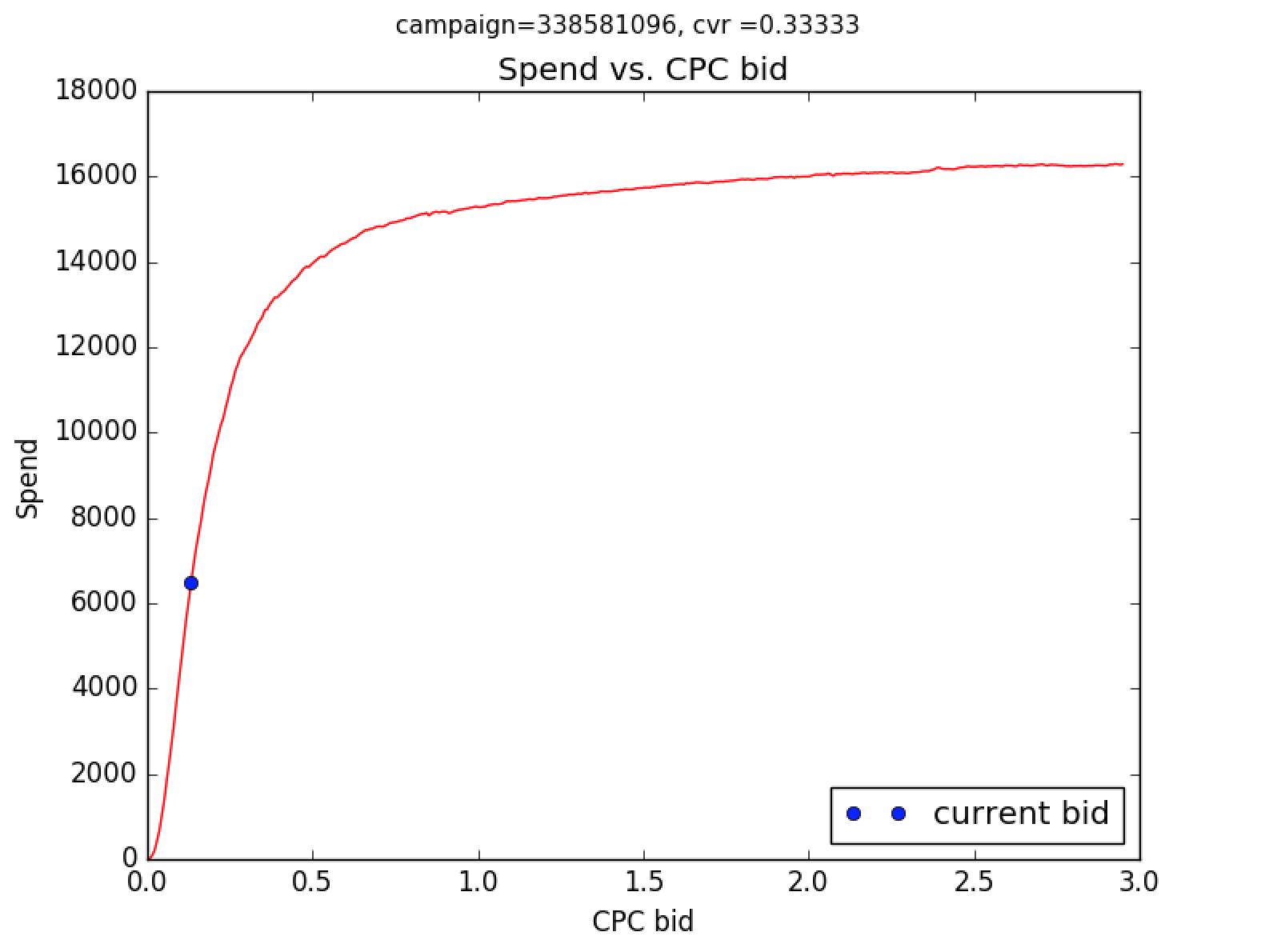}
		\caption{Ads 3: spend vs. bid}
	\label{fig:spend-2}
	\end{subfigure}
	\begin{subfigure}[b]{0.33\textwidth}		
	\centering
	\includegraphics[height=2in, width=\linewidth]{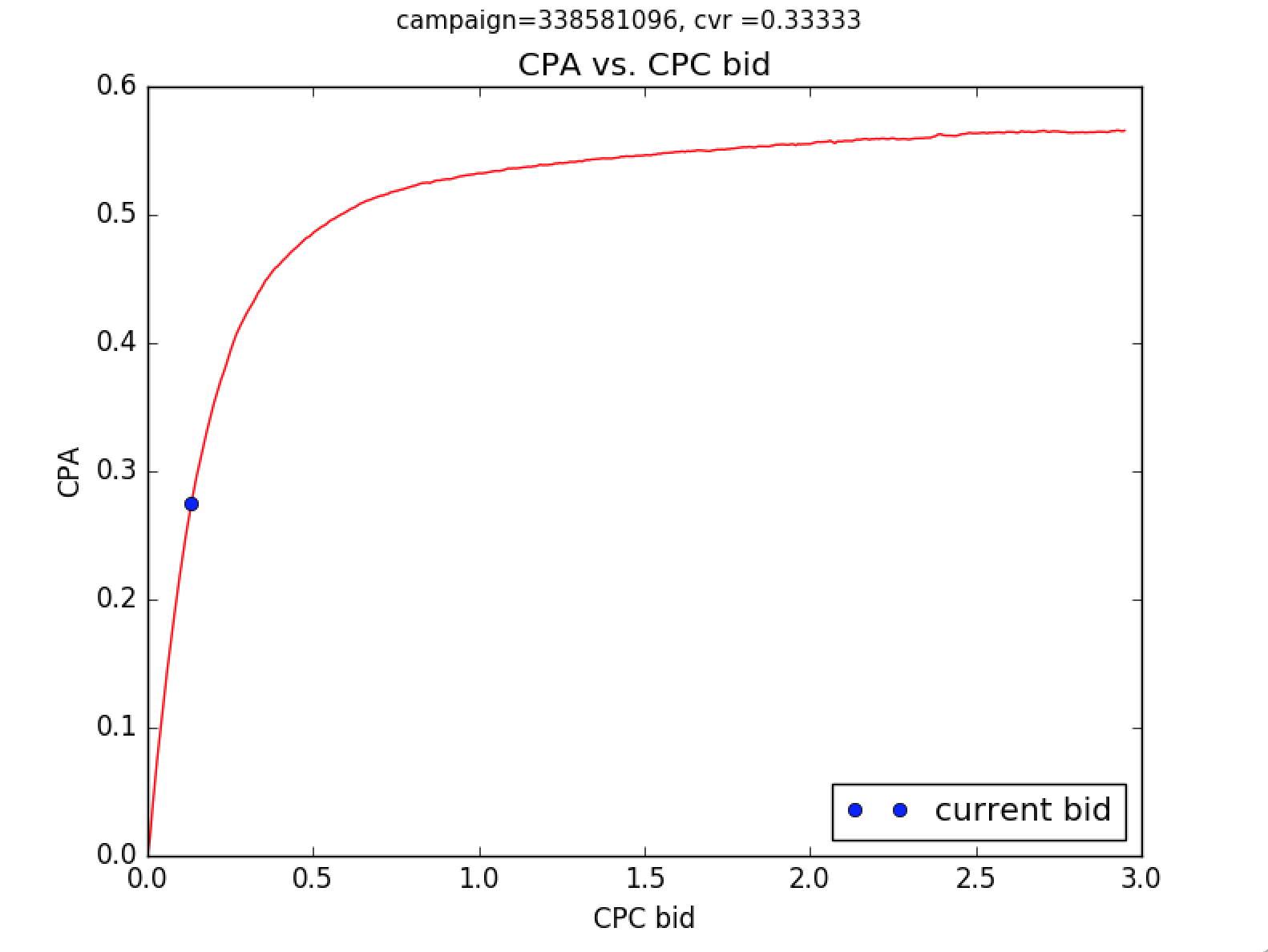}
	\caption{Ads 3: cpa vs. bid}
	\label{fig:cpa-3}
	\end{subfigure}	
	\caption{\small {CPA vs. bid,  Click vs. bid, Spend vs. bid curves for several different ads campaigns. CVR for each ads campaigns: 
	ads 1: $0.0270$, ads 2: $0.0556$, ads 3: $0.333$.
	}}
	\label{fig:cpa-click-spend}
\end{figure*}

\subsubsection{Bid landscape model performance}
We therefore first obtain the ground-truth of them:  $
{ win\_rate(bid) = \frac{click(bid)}{Impression \times CTR},}
\nonumber \\
{ eCPM\_cost(bid) = \frac{ spend(bid)}{click(bid)} \times CTR \times 1000,}
$
where $bid$ is set to the current bid, $click(bid), spend(bid), Impression,  CTR$ are obtained from the true observations of ads campaign. The forecast win-rate (denoted as $\hat{winrate}$) and eCPM (denoted as $\hat{eCPM}$) distributions are computed using the proposed algorithm. In evaluation, we evaluate their performance on the current bid price, and compare the difference between the ground-truth (e.g., $eCPM$)
and the corresponding forecasted values (e.g., $\hat{eCPM}$).  Similar to CPA goal evaluation, we use the mean absolute percentage error (MAPE) and root mean square percentage error (RMSPE) of Eq.(\ref{EQ:ape}) as the measurements.  The smaller the better. 
For win-rate model, we compare the following baseline methods:

{
\begin{table}
\caption{\small Performance comparisons of MAPE, RMSPE between true winrate and predicted $\hat{winrate}$ using (1) proposed method; (2) flat-rate (=0.1); (3)  flat-rate (=0.2); (4) flat-rate (=0.3); 
(4)  survival model; (5) log-normal distribution.}
\label{tbl:winrate_ape}
\begin{center}
\begin{tabular}{ p{0.8cm}|p{0.85cm}|p{1.05cm}|p{1.05cm}|p{1.05cm}|p{0.8cm}|p{0.8cm}}
\hline
\hline
Method & Our method & flat rate (=0.1)  & flat rate (=0.2) & flat rate (=0.3) & survival model &  log-normal \\
\hline
MAPE & 20.07\% & 40.19\% & 33.49\% & 35.27\% & 31.10\% & 28.73\% \\
 RMSPE & 28.75\% & 61.70\% & 50.48\% & 47.66\% & 41.50\% & 35.48\% \\
\hline
\hline
\end{tabular}
\end{center}
\end{table}
}
{
\begin{table}
\caption{\small Performance comparisons of MAPE, RMSPE between true eCPM cost and predicted $\hat{eCPM}$ using (1) proposed method; (2) flat-rate (=0.8); (3)  flat-rate (=0.9); 
(4)  survival model; (5) log-normal distribution.}
\label{tbl:ecpm_ape}
\begin{center}
\begin{tabular}{ p{0.85cm}|p{0.85cm}|p{1.1cm}|p{1.1cm}|p{0.85cm}|p{0.8cm} }
\hline
\hline
Method & Our method & flat rate (=0.8) & flat rate (=0.9) & survival model & log-normal  \\
\hline
MAPE & 13.75\% &  21.43\% & 19.89\% & 25.43\% & 22.83\% \\
RMSPE & 18.41\%  & 33.57\% & 29.84\% & 35.18\% & 32.96\% \\
\hline
\hline
\end{tabular}
\end{center}
\end{table}
}

$\bullet$ {\bf Survival model} It models winning price~\cite{DBLP:conf/kdd/ZhangZWX16} during the investigation period from low to high as the patient's underlying survival period from low to high days.  For example, if bid $b$ wins the auction with the observed winning price $z$, then it is analogous to the observation of the patient's death on day $z$. On the other hand, if the bid $b$ loses the auction, and then this is analogous to the patient's left from the investigation on day $b$.   In more detail, 
the winning probability $b_j$ is the bid price at auction $j$, $d_j$ denotes the number of auction winning cases with winning price $b_j-1$, and $n_j$ is the number of auction losing cases with bid price
$b_j -1$ at price $b_x$, which is equal to
\begin{eqnarray}
win\_rate(b_x) = 1 - \prod_{b_j < b_x} \frac{n_j - d_j}{n_j}.
\label{EQ:survival}
\end{eqnarray}

\newcommand{\E}{\mathbb{E}}
\newcommand{\Var}{\mathrm{Var}}

$\bullet$ {\bf Flat win-rate} Empirically, we set win-rate the same for all different bid prices. In particular, we list the results with win-rate in set \{0.10, 0.20, 0.30\}.

$\bullet$ {\bf Log-normal estimation} We use log-normal distribution~\cite{DBLP:conf/kdd/CuiZLM11} to model the winning price distribution\footnote{
$\mu, \sigma$ are sample mean and variance, and the winning rate at bid $b$  is given by $winrate(b) = \int^b_{\infty} Pr(z) dz $.} $Pr(z) \sim \text{logNorma}l(\mu, \sigma)$.
The baselines for eCPM\_cost evaluation are listed as follows. 

$\bullet$  {\bf Survival model} The idea is to obtain the winning price as the true cost as introduced in~\cite{DBLP:conf/kdd/ZhangZWX16}  based on non-parametric estimation using  the observed impressions and the losing bid requests.

$\bullet$  {\bf Flat cost} Intuitively, we set the CPC cost is equal  to 80\% and 90\% of CPC bid, and then compute the corresponding eCPM\_cost.  This is based on the simple observation: most campaigns win the auctions with price slightly slower than the bidding price. 

$\bullet$ {\bf Log-normal estimation} The idea is to get all eCPM\_cost distributions at different bid price $bid$, and then use log-normal distribution~\cite{DBLP:conf/kdd/CuiZLM11} to model each eCPM\_cost distribution with sample mean $\E(cost)$ and variance $Var(cost)$.  Finally return the $\E(cost)$ as the predicted cost for the given bid price.

From Tables~\ref{tbl:winrate_ape}, ~\ref{tbl:ecpm_ape}, we observe that our method outperforms the other baselines for both winrate and eCPM cost evaluation.  For the win-rate model, survival model does not perform as well as we expect since the auctions actually do not perform the same way as the death of the patient, and the simple analogy to survival process is not accurate. Moreover, there are many losing campaigns without any log information, which further biases the estimation results. 

Log-normal model performs reasonable well because it captures the ``average'' behaviors of auctions, and the result looks reasonable for many ad campaigns although there are deviations from aggregation-level mean. Our approach actually, provides more flexibility to model the \texttt{c.d.f} of cost and win-rate distributions without any prior distribution assumption\footnote{Actually, we address data sparsity problem by considering all the available positions in auctions, and address the uncertainty issue using density estimation as is  validated in experiment without any counterfactuals.}. It is surprising the even flat-rate with ratio=0.9 eCPM\_cost model performs reasonable well. The reason is that cost, empirically, is generally slightly slower than the bid, and therefore 90\% of bidding price is a good candidate for ``robust"  estimation.

\subsection{Online A/B test}

We apply A/B test to the ad traffic using two scenarios (1) without any bid optimization, 5\% traffic with the default behavior; (2) apply CPA goal optimization, 5\% traffic with the new changes in live experiment. Notice that different categories of advertisers may have quite different CPA goals, and therefore instead of setting a uniform CPA goal, we set CPA goal corresponding to the bid that uses up the budget.  
In particular, we look at the \emph{bid increase ratio},  \emph{click increase ratio} and \emph{ROI increase ratio}. 
Bid increase ratio (BIR) measures the percentage of bid increase from current bid ($bid^i_{cur}$ for ad $i$) to recommended bid ($bid^i_{cpa}$ for ad $i$) using proposed CPA model,  {\it i.e.,}
\begin{eqnarray}
BIR = \sum_{i} w_i \frac{bid^i_{cpa} - bid^i_{cur}}{bid^i_{cur}},
\label{EQ:BIR}
\end{eqnarray}
where $w_i$ is the normalized weight (i.e., $\sum_i w_i = 1$) for campaign $i$ which is proportional to the campaign spend. 
Click increase ratio (CIR) and ROI increase ratio (RIR) measure the percentage of click and ROI increase from current bid to recommended bid as defined in BIR,  respectively. 
The results are shown below in Table~\ref{table:ABtest} over all slices (e.g., desktop, android, IOS) and campaigns (including budgeted and non-budgeted campaigns). We observed the clear performance gains using the proposed bid optimization strategy over baselines. 

Even if every campaign applies the bid optimization based on suggestion, then everyone will adjust their bids based on other's bids. As a result, the winning rate distribution estimated by historical data is likely to change. Fortunately, the bid landscape model is also daily updated by incorporating the changes of advertisers' bidding behaviors by giving enough emphasize on the recent behaviors. The model is gradually updated to accommodate the new changes. 

{
\begin{table}
\caption{Online Test Result: performance liftup over baseline.}
\begin{tabular}{ |c|c|c|c|}
\hline
Measurement & BIR& CIR & RIR \\
\hline
Performance liftup &+12.81\% & +25.76\% & +25.37\%\\
\hline
\end{tabular}
\label{table:ABtest}
\end{table}
}

\section{Conclusion}
\label{sec:conclude}

This paper presents a novel way for bid optimization that balances CPA goals, conversions and spend. The experiment numbers demonstrate our method is effective in getting more conversions with improved ROI for the same money and better in utilizing the budget. 
We would like to optimize campaign objectives (e.g., post-app install, post-app event, post-app value, post-app user engagement) that advertisers really care about in practice.


\newpage
{
\bibliographystyle{ACM-Reference-Format}
\bibliography{cpaall} 
}

\end{document}